\pdfoutput=1
\documentclass[manuscript, nonacm]{acmart}

\AtBeginDocument{%
  \providecommand\BibTeX{{%
    \normalfont B\kern-0.5em{\scshape i\kern-0.25em b}\kern-0.8em\TeX}}}

\setcopyright{acmcopyright}
\copyrightyear{2018}
\acmYear{2018}
\acmDOI{10.1145/1122445.1122456}

\acmConference[Woodstock '18]{Woodstock '18: ACM Symposium on Neural
  Gaze Detection}{June 03--05, 2018}{Woodstock, NY}
\acmBooktitle{Woodstock '18: ACM Symposium on Neural Gaze Detection,
  June 03--05, 2018, Woodstock, NY}
\acmPrice{15.00}
\acmISBN{978-1-4503-XXXX-X/18/06}

\usepackage[utf8]{inputenc}
\usepackage{listings}
\usepackage{amsmath}
\usepackage[skins, breakable]{tcolorbox}
\usepackage{multirow}
\usepackage{amsthm}
\usepackage{amsfonts}
\usepackage{graphicx}
\usepackage{appendix}
\usepackage{xspace}
\usepackage{color}
\usepackage{tikz}
\usepackage{pgfplots}
\usepackage{textcomp}
\usepackage{enumitem}
\usepackage{float}
\usepackage{hhline}
\usepackage{bm}

\newtheorem{lemma}{Lemma}
\newtheorem{theorem}{Theorem}
\newtheorem{definition}{Definition}

\newcommand{\fc} {\mathcal{C}}

\newcommand{\ff} {\mathcal{F}}

\newcommand{\fs} {\mathcal{S}}

\newcommand{\propose}{\texttt{propose}\xspace}

\newcommand{\vote}{\texttt{vote}\xspace}

\newcommand{\lock}{\texttt{lock}\xspace}
\newcommand{\rank}{\texttt{rank}\xspace}

\newcommand{\latency} {good-case latency\xspace}
\newcommand{\clockskew}{\sigma}
\newcommand{\clone}{round-$0$\xspace}
\newcommand{\cltwo}{round-$1$\xspace}
\newcommand{\clthree}{round-$2$\xspace}
\newcommand{\timeout}{\texttt{timeout}\xspace}
\newcommand{\status}{\texttt{status}\xspace}

\newcommand{\vbblong}{{partially synchronous validated Byzantine broadcast}\xspace}
\newcommand{\vbbshort}{{psync-VBB}\xspace}
\newcommand{\psyncbblong}{{partially synchronous Byzantine broadcast}\xspace}
\newcommand{\psyncbbshort}{{psync-BB}\xspace}

\newtcolorbox{mybox}[1][]{
enhanced,
colback=white,
boxsep=0pt,
#1
} 

\newcommand{\stitle}[1]{\vspace{0.5ex} \noindent\textsf{\textbf{#1}}}

\renewcommand{\paragraph}[1]{\smallskip\stitle{#1}}

\newcommand{\daniel}[1]{{}}

\begin{document}

\title{Good-case Latency of Byzantine Broadcast: A Complete Categorization}

\author{Ittai Abraham}
\affiliation{%
  \institution{VMware Research}
  \country{Israel}
}
\email{iabraham@vmware.com}

\author{Kartik Nayak}
\affiliation{%
  \institution{Duke University}
  \country{USA}
}
\email{kartik@cs.duke.edu}

\author{Ling Ren}
\affiliation{%
  \institution{University of Illinois at Urbana-Champaign}
  \country{USA}
}
\email{renling@illinois.com}

\author{Zhuolun Xiang}
\affiliation{%
  \institution{University of Illinois at Urbana-Champaign}
  \country{USA}
}
\email{xiangzl@illinois.com}


\begin{abstract}
  This paper explores the problem {\em good-case latency} of Byzantine fault-tolerant broadcast, motivated by the real-world latency and performance of practical state machine replication protocols.
  The good-case latency measures the time it takes for all non-faulty parties to commit when the designated broadcaster is non-faulty. 
  We provide a complete characterization of tight bounds on good-case latency, in the authenticated setting under synchrony, partial synchrony and asynchrony.
  Some of our new results may be surprising, e.g., 2-round PBFT-style partially synchronous Byzantine broadcast is possible if and only if $n\geq 5f-1$, and a tight bound for good-case latency under $n/3<f<n/2$ under synchrony is not an integer multiple of the delay bound.  
  
\end{abstract}





\maketitle

\section{Introduction}
\label{sec:intro}


\begin{table*}[t]
\centering
\setlength\doublerulesep{0.5pt}
\begin{tabular}{|c|c|c|c|c|}
\hline
\textbf{Problem} & \textbf{Timing Model} & \textbf{Resilience} & \textbf{Lower Bound} &
\textbf{Upper Bound} \\ \hhline{=====} 
BRB & Asynchrony & $n\geq 3f+1$ & $2$ rounds 
& $2$ rounds 
\\
 \hhline{=====}
\multirow{2}{*}{Psync-BB} & \multirow{2}{*}{Partial Synchrony} & $n\geq 5f-1$ & $2$ rounds 
& \textbf{$\bm{2}$ rounds} \\ \cline{3-5} 
& & $3f+1\leq n\leq 5f-2$ & \textbf{$\bm{3}$ rounds} & $3$ rounds~\cite{castro1999practical} \\ 
 \hhline{=====}
 \multirow{5}{*}{BB} &
\multirow{5}{*}{Synchrony} & $0<f < n/3$ & ${2\delta}$ & $\bm{ 2 \delta}$ \\ \cline{3-5} 
 & & $f=n/3$ & $\bm{\Delta+\delta}$ & $\bm{\Delta+\delta}$ \\ \cline{3-5} 
 & & \multirow{2}{*}{$n/3< f <n/2$} & \multirow{2}{*}{\textbf{\begin{tabular}[c]{@{}c@{}}sync start $\bm{\Delta+\delta}$\\ unsync start $\bm{\Delta+1.5\delta}$\end{tabular}}} & \multirow{2}{*}{\textbf{\begin{tabular}[c]{@{}c@{}}sync start $\bm{\Delta+\delta}$\\ unsync start $\bm{\Delta+1.5\delta}$\end{tabular}}} \\
 &  &  & & \\ \cline{3-5} 
 & & $n/2 \leq f<n$ & {$\bm{(\lfloor \frac{n}{n-f} \rfloor -1)\Delta}$} & $O(\frac{n}{n-f})\Delta$~\cite{constantroundBB} \\ 
 \hline
\end{tabular}
\smallskip 
\caption{Upper and lower bounds for \latency of Byzantine fault-tolerant broadcast. \textbf{Our new (and nontrivial) results are marked bold.} 
Under asynchrony, all lower and upper bounds are for BRB.
Under partial synchrony, all lower and upper bound are for a new broadcast formulation named \psyncbblong we propose (Definition~\ref{def:psyncbb}).
To strengthen all results under synchrony, all lower bounds are for BRB and synchronized start, and all upper bounds are for BB and unsynchronized start, unless otherwise specified. }
\vspace{-1em}
\label{table:results}
\end{table*}

Byzantine fault-tolerant broadcast is a fundamental problem in distributed computing.
In Byzantine broadcast (BB) or Byzantine reliable broadcast (BRB), there is a designated broadcaster that sends its input value to all parties, and all non-faulty parties must deliver the same value. Moreover, if the broadcaster is non-faulty, then the delivered value must be broadcaster's input. 
BB requires all non-faulty parties to eventually terminate, while BRB relaxes the condition to only require termination when the broadcaster is honest or if a non-faulty party terminates.

One of the most important practical applications of broadcast is to implement \textit{Byzantine fault-tolerant state machine replication} (BFT SMR), which ensures all non-faulty replicas agree on the same sequence of client inputs to provide the client with the illusion of a single non-faulty server.
Most of the practical solutions for BFT SMR are based on the Primary-Backup paradigm. In this approach, in each \emph{view}, one replica is designated to be the leader, and is in charge of a view to drive decisions, until replaced by the leader of the next view due to malicious behavior or bad network connection. 
The Primary-Backup approach for SMR exposes deep connections to broadcast. 
Each view in BFT SMR is similar to an instance of broadcast where with the leader taking on a similar role as the broadcaster\footnote{Each view of BFT SMR does not require committing an honest leader's proposed value, and thus, is weaker than BB and BRB.}, and hence an efficient broadcast protocol can be converted to an SMR protocol with similar efficiency guarantees.
Due to the importance of BFT SMR and the recent interest in permissioned blockchains, improving the latency of the BFT SMR and understanding its fundamental latency limit have been a research focus for several decades~\cite{castro1999practical, kotla2007zyzzyva, martin2006fast, gueta2019sbft, synchotstuff, abraham2020optimal}.

However, 
there exists a 
mismatch between theoretical studies on broadcast latency and practical SMR systems.
Most of the theoretical studies focused on the {\em worst-case latency} of broadcast. 
For Byzantine broadcast, the worst-case number of rounds required is $f+1$ to tolerate $f$ faults \cite{fischer1982lower}. 
As $f$ is typically assumed to be linear in $n$, any BB protocol will inevitably have a poor worst-case latency as $n$ increases.
However, in contrast, practical BFT SMR systems care more about the \emph{good-case}, in which a stable non-faulty leader stays in charge and drives consensus on many decisions. 
Another relatively minor disconnect lies in the ``life cycle'' of the protocol.
Broadcast and reliable broadcast require all parties to \emph{halt} or \emph{terminate} after agreeing on a single value, while practical SMR protocols are intended to run forever; replicas \emph{commit} or \emph{decide} on an ever-growing sequence of values.
Hence, in contrast to the worst-case latency to halt, we argue the importance of \emph{\latency to commit} for broadcast protocols defined as follows.

\begin{definition}[Informal; Good-case Latency]\label{def:goodcase:sync}
    The \latency of a broadcast protocol is the time for all honest parties to commit (over all executions and adversarial strategies), given the designated broadcaster is honest.
\end{definition}

In fact, practical systems often implicitly talk about the good-case latency without formally defining it as such. In a talk from 2000~\cite{Barbara2001}, Barbara Liskov commented that she did not know whether 3 rounds are optimal for PBFT~\cite{castro1999practical}. 
It is clear that Liskov was implicitly referring to the number of rounds needed to reach agreement in the good case that the primary is non-faulty.

Our work's main contribution is a formal theoretical framework to address exactly this question and be able to prove the minimum number of phases in the good-case where the primary is non-faulty. In fact, our work gives a complete and tight categorization on \latency, for any threshold adversary size, both for the synchronous model, partially synchronous model and for the asynchronous model in the authenticated setting (i.e., with signatures). 
Another contribution of our work is a new broadcast formulation, named \psyncbblong (\psyncbbshort), for the partially synchronous model to better abstract a single shot of BFT SMR such as PBFT, since no existing broadcast formulation captures the properties of partially synchronous BFT SMR protocols such as PBFT. For instance, most existing BFT SMR solutions are leader-based and will replace an honest leader during asynchrony, while Byzantine broadcast and Byzantine reliable broadcast force parties to commit an honest broadcaster's value.
We summarize our findings in Table~\ref{table:results} and some of them may be quite surprising.

\paragraph{Complete categorization for \latency in asynchrony.}
The results for BRB under asynchrony turn out to be strightforward.
Although by definition BRB can always be solved with \latency of $2$ rounds, the results are not useful for solving BFT SMR, since the parties in BRB are allowed to never commit when the broadcaster is Byzantine.

\begin{theorem}[Informal; tight bounds on \latency in asynchrony]
For Byzantine Reliable Broadcast with $f$ Byzantine parties in the asynchronous and authenticated setting, in the good-case, 
2 rounds are necessary and sufficient iff $\, n \geq 3f+1 $ (Section~\ref{sec:async}).
\end{theorem}

\daniel{check here}

\paragraph{Complete categorization for \latency in partial synchrony.}
PBFT~\cite{castro1999practical} shows that in the good case, 3 rounds are sufficient with $n \geq 3f+1$ and FaB~\cite{martin2006fast} shows that 2 rounds are sufficient with $n \geq 5f+1$. Somewhat surprisingly, we show that these results are not tight in the authenticated model (which PBFT and all follow-up work assume), and the boundary between 2 and 3 rounds is at $n \geq 5f-1$. 
We propose a new broadcast formulation named \psyncbblong (\psyncbbshort) that abstracts a single-shot of BFT SMR under partial synchrony, and show the following results. 

\begin{theorem}[Informal; tight bounds on \latency in partial synchrony]
For Partially Synchronous Byzantine Broadcast with $f$ Byzantine parties in the partially synchronous and authenticated setting, in the good-case: 
\begin{enumerate}
\item 2 rounds are necessary and sufficient if $\, n \geq 5f-1 $ (Section~\ref{sec:psync:1}), and   
\item 3 rounds are necessary and sufficient if $\, 3f+1 \leq n < 5f-1$ (Section~\ref{sec:psync:2}).
\end{enumerate}
\end{theorem}

The new 2-round good-case \psyncbbshort protocol solves a single shot of BFT SMR within $2$ rounds in the good case.
We extend the protocol to obtain a practical BFT SMR in our complementary paper~\cite{abraham2021fast}. 
Observe one interesting and important special case here: when $f=1$, we have $n=4=3f+1=5f-1$, and 3-round PBFT is not optimal, as 2 rounds are sufficient.

\paragraph{Complete categorization for \latency in synchrony.} 
We give a complete categorization for \latency in the synchronous model, where message delays are bounded by a known upper bound $\Delta$. 
There turns out to be a surprisingly rich spectrum here. 

We adopt the separation between the conservative \textit{worst-case} bound $\Delta$ and the \textit{actual} (unknown) bound $\delta \leq \Delta$ as suggested in \cite{ierzberg1989efficient, pass2017hybrid, abraham2020optimal, abraham2020brief, shrestha2020optimality}. 
Moreover, our categorization highlights the importance of the assumption on the synchronization of when the protocol starts at each party. We distinguish two models: the \textit{synchronized start} model assumes all parties start the protocol at the exact same time; the \textit{unsynchronized start} model assumes all parties start the protocol within a known time interval bounded by the clock skew (see Section~\ref{sec:prelim}).
Whenever applicable, we prove lower bounds in the synchronized start model (hence they also apply to unsynchronized start), and  upper bounds in the unsynchronized start model (hence they also apply to synchronized start). 
The case of $n/3<f<n/2$ is the only exception, where the tight bounds differ based on the synchronization assumption.
It is also worth highlighting that, somewhat surprisingly, the bound for $n/3<f<n/2$ with an unsynchronized start is $\Delta+1.5\delta$, which is {\em not} an integer multiple of the message delay.

\begin{theorem}[Informal; bounds on \latency in synchrony]
For Byzantine Broadcast and Byzantine Reliable Broadcast with $f$ Byzantine parties in the synchronous and authenticated setting, in the good-case: 
\begin{enumerate}
\item if $~0<f<n/3$, then $2\delta$ is necessary and sufficient (Section~\ref{sec:sync:1});

\item if $~f=n/3$, then $\Delta+\delta$ is necessary and sufficient (Section~\ref{sec:sync:2});

\item if $~n/3<f<n/2$ then $\Delta+\delta$ is necessary and sufficient in the synchronized start model (Section~\ref{sec:sync:syncstart}), and $\Delta+1.5\delta$ is necessary and sufficient in the unsynchronized start model (Section~\ref{sec:sync:unsyncstart}); 

\item if $n/2\leq f<n$ then $(\lfloor \frac{n}{n-f} \rfloor -1)\Delta$ is necessary, and $O(\frac{n}{n-f})\Delta$ is sufficient (Section~\ref{sec:sync:5}).
\end{enumerate}
\end{theorem}

\section{Preliminaries}
\label{sec:prelim}

\paragraph{Model of execution.} 
We define a protocol for a set of $n$ parties, among which at most $f$ are Byzantine faulty and can behave arbitrarily. If a party remains non-faulty for the entire protocol execution, we call the party honest.
During an execution $E$ of a protocol, parties perform sequences of events, including {\em send, receive/deliver, local computation}. 
The {\em local history} at some party $i$ during an execution $E$ refers to the initial state and the sequence of events performed by party $i$ in $E$, denoted as $E|_i$.
An honest party $i$ cannot distinguish two executions $E_1$ and $E_2$ if its local history is identical in both executions, i.e., $E_1|_i=E_2|_i$.
If the protocol is deterministic, for any two executions, if an honest party has the same initial state and receives the same set of messages at the same corresponding time points (by its local clock), the honest party will have the same local history and thus cannot distinguish two executions.
We will use the standard indistinguishability argument to prove our lower bounds.
In this paper, we investigate results for deterministic authenticated protocols.
We use (perfect) digital signatures and public-key infrastructure (PKI), assume ideal unforgeability, and use $\langle m \rangle_i$ to denote a message $m$ signed by party $i$. We call any message valid, if the message is in the correct format and properly signed by the corresponding party (sender). We say a party detects equivocation if it receives messages containing different values signed by the broadcaster.

\paragraph{Synchrony, partial synchrony, and asynchrony.}
We consider three standard network models, synchrony, partial synchrony and asynchrony. Under \textit{asynchrony}, the adversary can control the message delay of any message to be an arbitrary non-negative value.
Under \textit{partial synchrony}, the adversary can control the message delay of any message to be an arbitrary non-negative value until a Global Stable Time (GST), after which the message delays are bounded by $\Delta$.

For a more accurate latency characterization under \textit{synchrony}, we follow the literature~\cite{ierzberg1989efficient, pass2017hybrid, synchotstuff} to separate the {\em actual bound $\delta$}, and the {\em conservative bound $\Delta$} on the network delay:
\begin{itemize}[itemsep=0pt,topsep=0pt]
    \item For one execution, $\delta$ is the upper bound for message delays between any pair of honest parties, but the value of $\delta$ is {\em unknown} to the protocol designer or any party. Different executions may have different $\delta$ values. 
    \item For all executions, $\Delta$ is the upper bound for message delays between any pair of honest parties, and the value of $\Delta$ is {\em known} to the protocol designer and all parties. 
\end{itemize}
In other words, $\Delta$ is the maximum network delay bound assumed by the synchronous model, and by definition $\delta\leq \Delta$. 
In practice, the parameter $\Delta$ is usually chosen conservative, and thus $\delta\ll\Delta$.
In this model, for any execution with an actual message delay bound $\delta$, the adversary can control the delay of any message between two honest parties to be any value in $[0, \delta]$.
Between a pair of parties where at least one is Byzantine, the adversary can control the message delay to be any non-negative value or even infinity.
This can be easily achieved by having Byzantine parties postpone sending or reading the message to simulate an arbitrary delay.
A simulated delay of $\infty$ means that a message is never sent or is discarded without being read. 

\paragraph{Clock synchronization.}
For synchronous protocols, 
each party is equipped with a local clock that starts counting at the beginning of the protocol execution. 
We assume the \emph{clock skew} is at most $\clockskew$, i.e., they start the protocol at most $\clockskew$ apart from each other.
We assume parties have {\em no clock drift} for convenience. 
There exist clock synchronization protocols~\cite{dolev1995dynamic, abraham2019synchronous} that guarantee a bounded clock skew of $\clockskew\leq \delta$.
At the same time, there is a negative result that shows clocks cannot be synchronized within $(1-1/n)\delta$~\cite{attiya2004distributed}, hence the clock skew must $\geq (1-1/n)\delta\geq 0.5\delta$ in the nontrivial case of $n\geq 2$ for broadcast.
If $\clockskew=0$, all parties start their local clock at the same time, and we refer to such model as the {\em  synchronized start} model. Otherwise, it is called the {\em unsynchronized start} model. 
To strengthen our results for the synchronous case, {\em all lower bound results assume the synchronized start model and all upper bound results assume the unsynchronized start model, unless otherwise specified}.
Since the value of $\delta$ is unknown to the protocol designer or any party, our upper bound results will use $\Delta$ as the parameter for clock skew in the protocol. Note that the actual clock skew is still $\clockskew\leq \delta$, guaranteed by the clock synchronization protocols~\cite{dolev1995dynamic, abraham2019synchronous}.
For the lower bound under unsynchronized start, we assume the smallest achievable clock skew $\clockskew=0.5\delta$ to strengthen the bound.
Partially synchronous protocols use local blocks with {\em no clock drift and arbitrary clock skew}, for timeout and view-change. \daniel{check here}
Asynchronous protocols do not use clocks and make no assumption on clocks. 

\paragraph{Byzantine broadcast variants.}
We investigate two standard variants of Byzantine broadcast problem for synchrony and asynchrony, and define a new variants of Byzantine broadcast problem for partial synchrony (Definition~\ref{def:psyncbb}).

\begin{definition}[Byzantine Broadcast (BB)]
\label{def:bb}
A Byzantine broadcast protocol must satisfy the following properties.
    \begin{itemize}[itemsep=0pt,topsep=0pt]
        \item Agreement. If two honest parties commit values $v$ and $v'$ respectively, then $v=v'$.
        \item Validity. If the designated broadcaster is honest, then all honest parties commit the broadcaster's value and terminate.
        \item Termination. All honest parties commit and terminate.
    \end{itemize}
\end{definition}

\begin{definition}[Byzantine Reliable Broadcast (BRB)]
\label{def:bb}
A Byzantine reliable broadcast protocol must satisfy the following properties.
    \begin{itemize}[itemsep=0pt,topsep=0pt]
        \item Agreement. Same as above.
        \item Validity. Same as above.
        \item Termination. If an honest party commits a value and terminates, then all honest parties commit a value and terminate.
    \end{itemize}
\end{definition}

For partial synchrony, we define {\em \psyncbblong} below, which abstracts the single-shot of existing partially synchronous BFT SMR protocols such as PBFT~\cite{castro1999practical}.

\begin{definition}[\psyncbblong (\psyncbbshort)]
\label{def:psyncbb}
A \psyncbblong protocol provides the following properties.
    \begin{itemize}[itemsep=0pt,topsep=0pt]
        \item Agreement. Same as above.
        
        \item Validity. If the designated broadcaster is honest and $GST=0$, then all honest replicas commit the broadcaster's value.
        
        \item Termination. All honest replicas commit and terminate after $GST$.
    \end{itemize}
\end{definition}

In comparison, BB is a harder than both \psyncbbshort and BRB, as BRB relaxes the termination property of BB to allow either all honest parties commit or no honest party commits, and \psyncbbshort relaxes both termination and validity.
Also note that under synchrony, \psyncbbshort is the same as BB, and BB can only be solved under synchrony.
As a result, under synchrony, any protocol that solves BB (same as \psyncbbshort) also solves BRB, and thus {\em any upper bound result (including \latency) for BB also implies the same upper bound result for BRB}.
Moreover, {\em  any lower bound result for BRB also implies the same lower bound result for BB}.
BRB is not comparable to \psyncbbshort, as the validity of BRB is stronger (honest parties must commit the honest broadcaster's value in BRB but not in \psyncbbshort), while the termination of \psyncbbshort is stronger (all honest parties commits after GST in \psyncbbshort but not in BRB).

In this paper, we present all the upper and lower bounds in the strongest possible form for synchrony and asynchrony as follows.
We present upper bounds (i.e., construct protocols) for BB under synchrony (which equal \psyncbbshort and also solve BRB), BRB under asynchrony as BB is impossible to solve even with a single fault under asynchrony.
We present all lower bound results for BRB for both synchrony and asynchrony (which also apply to BB).

For partial synchrony, we present all lower bounds for \psyncbbshort, and upper bounds for a slightly stronger formulation that additionally requires the committed values to be {\em externally valid} for some external predicate $\ff: \{0,1\}^l\rightarrow \{true, false\}$, to better capture the partially synchronus BFT SMR like PBFT in practice.

\begin{definition}[\vbblong (\vbbshort)]
\label{def:psyncbb}
A \vbblong protocol provides the following properties.
    \begin{itemize}[itemsep=0pt,topsep=0pt]
        \item Agreement. Same as above.
        
        \item Validity. If the designated broadcaster is honest and $GST=0$, then all honest replicas commit the broadcaster's value;
        otherwise the value $v$ committed by any honest replica satisfies $\ff(v)=true$ ($v$ is externally valid).
        
        \item Termination. Same as above.
    \end{itemize}
\end{definition}

A \vbbshort protocol solves \psyncbbshort protocol by definition, and directly solves a single-shot of BFT SMR.
Hence, we will present upper bounds for \vbbshort, and the extension to a practical BFT SMR protocol can be found in our complementary paper~\cite{abraham2021fast}.

\daniel{check above}

We will also use {\em Byzantine agreement} as a primitive to simplify the construction of our BB protocols under synchrony. 
The Byzantine agreement gives each party an input, and its validity requires that if all honest parties have the same input value, then all honest parties commit that value.
In addition, due to clock skew, in our synchronous BB protocols, the honest parties may invoke the BA at times at most $\clockskew$ apart from each other. 
Therefore, we need the BA primitive to tolerate up to $\clockskew$ clock skew.
For instance, any synchronous lock-step BA can do so by using a clock synchronization algorithm \cite{dolev1995dynamic, abraham2019synchronous} to ensure at most $\Delta$ clock skew, and setting each round duration to be $2\Delta$ to enforce the abstraction of lock-step rounds. 
Our synchronous BB protocols in Figure~\ref{fig:sync:3f},  \ref{fig:sync:D+d} and \ref{fig:sync:D+1.5d}, \ref{fig:sync:2d} will use such a BA primitive.

\paragraph{Good-case latency of broadcast.}
As explained in Section~\ref{sec:intro}, improving the latency performance of BFT SMR protocols motivates our investigation on the \latency of the family of Byzantine fault-tolerant broadcast protocols.
Depending on the network model, the measurement of latency is different. 

\begin{definition}[Good-case Latency under Synchrony]\label{def:goodcase:sync}
    A Byzantine broadcast (or Byzantine reliable broadcast) protocol has \latency of $~T$ under synchrony, if all honest parties commit within time $~T$ since the broadcaster starts the protocol (over all executions and adversarial strategies), given the designated broadcaster is honest.
\end{definition}

To measure the latency of a partially synchronous protocol, we use the natural notion of {\em synchronous rounds}, following partially synchronous protocols like PBFT~\cite{castro1999practical}.

\begin{definition}[Good-case Latency under Partial Synchrony]\label{def:goodcase:psync}
    A \psyncbblong protocol has \latency of $R$ rounds under partial synchrony, if all honest parties commit within $R$ synchronous round (over all executions and adversarial strategies), given the designated broadcaster is honest and $GST=0$.
\end{definition}
\daniel{use sync rounds for psync}

To measure the latency of an asynchronous protocol, we adopt the standard and natural notion of {\em asynchronous rounds} from the literature~\cite{canetti1993fast}.
We defer the formal definitions to Appendix~\ref{app:asyncmodel}.

\begin{definition}[Good-case Latency under Asynchrony]\label{def:goodcase:async}
    A Byzantine reliable broadcast protocol has \latency of $R$ rounds under asynchrony, if all honest parties commit within asynchronous round $R$ (over all executions and adversarial strategies), given the designated broadcaster is honest.
\end{definition}

\section{Asynchronous Byzantine Fault-tolerant Broadcast}
\label{sec:async}
The standard broadcast formulation for asynchrony is Byzantine reliable broadcast, which is solvable if and only if $n\geq 3f+1$.
We show the {\em tight} lower and upper bound on the \latency of asynchronous BRB is $2$ rounds.

\paragraph{BRB lower bound $2$ rounds under $f>0$.}
This bound is almost trivial. The intuition is that if any BRB can guarantee a \latency of $1$ round, then parties must commit after receiving from the broadcaster. In an execution where the broadcaster is Byzantine and equivocates, this will lead to safety violations.

\begin{theorem}\label{thm:lb:async:2}
    Any Byzantine reliable broadcast protocol that is resilient to $f>0$ faults must have a \latency of at least $2$ rounds under asynchrony.
\end{theorem}

\begin{proof}[Proof of Theorem~\ref{thm:lb:async:2}.]
    Suppose there exists a BRB protocol $\Pi$ that has a \latency of $1$ round, which means the honest parties can always commit after receiving all round-$0$ messages but before receiving any round-$1$ messages, if the designated broadcaster is honest.
    Let party $s$ be the broadcaster, and divide the remaining $n-1$ parties into two groups $A,B$ each with $\geq 1$ party. 
    For brevity, we often use $A$ ($B$) to refer all the parties in $A$ ($B$).
    Consider the following three executions of $\Pi$. 
\begin{enumerate}
    \item Execution  1. The broadcaster $s$ is honest, and sends $0$ to all parties in round $0$.
    Since the broadcaster is honest, by validity and \latency, parties in $A,B$ will commit $0$ after receiving all round-$0$ messages but before receiving any round-$1$ messages.
    
    \item Execution  2. The broadcaster $s$ is honest, and sends $1$ to all parties in round $0$.
    Since the broadcaster is honest, by validity and \latency, parties in $A,B$ will commit $1$ after receiving all round-$0$ messages but before receiving any round-$1$ messages.
    
    \item Execution  3. The broadcaster $s$ is Byzantine, it sends $0$ to parties in $A$ and $1$ to parties in $B$ in round $0$.
\end{enumerate}

    {\em Contradiction.} 
    The set of round-$0$ messages received by $A$ from $B$ is identical in Execution 1 and 3 since the round-$0$ messages only depend on the initial states.
    Therefore, the parties in $A$ cannot distinguish Execution  $1$ and $3$ before receiving any round-$1$ message,  and thus will commit $0$ in Execution  $3$.
    Similarly, the parties in $B$ cannot distinguish Execution  $2$ and $3$ before receiving any round-$1$ message, and will commit $1$ in Execution  $3$.
    However, this violates the agreement property of BRB, and therefore no such protocol $\Pi$ exists.
\end{proof}

The same proof also applies to an even weaker broadcast formulation named Byzantine consistent broadcast (BCB), where termination of all honest parties is required only when the broadcaster is honest.

\paragraph{BRB upper bound $2$ rounds under $n\geq 3f+1$.}
\label{sec:async:ub:2r}
We show the tightness of the bound by presenting a trivial authenticated protocol $2$-round-BRB, which has \latency of $2$ rounds with $n\geq 3f+1$ parties, as presented in Figure~\ref{fig:async:2r}.
After the broadcaster proposes its value and parties send a \vote for the first valid proposal, each party waits for $n-f$ \vote messages for the same value to commit.

\begin{figure}[tb]
    \centering
    \begin{mybox}
\begin{enumerate}
    \item\label{rb1:step:propose} \textbf{Propose.} The designated broadcaster $L$  with input $v$ sends $\langle \texttt{propose}, v\rangle$ to all  parties.

    \item\label{rb1:step:vote1} \textbf{Vote.} 
    When receiving the first proposal $\langle \texttt{propose}, v\rangle$ from the broadcaster,
    send a \vote message for $v$ to all parties in the form of $\langle \vote, v \rangle_i$.

    \item\label{rb1:step:commit} \textbf{Commit.}
    When receiving $n-f$ signed \vote messages for $v$, forward these \vote messages to all other parties, commit $v$ and terminate.

\end{enumerate}
    \end{mybox}
    \vspace{-1em}
    \caption{$2$-round-BRB Protocol with $n\geq 3f+1$}
    \vspace{-1em}
    \label{fig:async:2r}
\end{figure}

\begin{theorem}\label{thm:ub:async}
    The $2$-round-BRB protocol solves Byzantine reliable broadcast with $n\geq 3f+1$ in the asynchronous authenticated setting, and has optimal good-case latency of $~2$ rounds.
\end{theorem}

\begin{proof}
    {\bf Agreement.}
    If any two honest parties commit different values at Step~\ref{rb1:step:commit}, then by standard quorum intersection argument, two sets of $n-f$ \vote messages must intersect at $\geq 2(n-f)-n\geq f+1$ parties, which implies some honest party sends \vote for different values, a contradiction.

    {\bf Validity and Good-case Latency.}
    If the broadcaster is honest, it sends the same proposal of value $v$ to all parties. Then all $n-f$ honest parties will multicast the \vote message for $v$. The Byzantine parties cannot make any honest party to commit a different value since $f<n-f$. All honest parties will eventually commit $v$ after receiving $n-f$ \vote messages at Step~\ref{rb1:step:commit} and terminate. The commit latency is $2$ rounds if the broadcaster is honest.
    
    {\bf Termination.}
    Suppose an honest party $h$ commits $v$ and terminates, its forwarded $n-f$ \vote messages for $v$ will eventually lead all honest parties to commit and terminate. 
\end{proof}

\section{Partially Synchronous Byzantine Fault-tolerant Broadcast}

In this section, we will present {\em tight} lower and upper bound results on the \latency of partially synchronous \psyncbblong under different resilience guarantees.
All the lower bound results are for authenticated \psyncbblong, and all the upper bound results are for authenticated \vbblong.

\subsection{$n\geq 5f-1$,  Matching Lower and Upper Bounds of $2$ Rounds}
\label{sec:psync:1}

\paragraph{\psyncbbshort lower bound $2$ rounds under $f>0$.}
Similar to the $2$-round lower bound for asynchronous BRB, this bound is also trivial and can be implied by a similar proof of Theorem~\ref{thm:lb:async:2}, which we will 
omit for brevity.

\begin{theorem}\label{thm:lb:psync:2}
    Any \psyncbblong protocol that is resilient to $f>0$ faults must have a \latency of at least $2$ rounds under partial synchrony.
\end{theorem}

\paragraph{\vbbshort upper bound $2$ rounds under $n\geq 5f-1$.}
\label{sec:psync:ub:2r}
In this section, we present an authenticated \vbblong protocol with \latency of $2$ rounds and only requires $n\geq 5f-1$ parties, shown in Figure~\ref{fig:psyncbb}.
The $(5f-1)$-\vbbshort protocol is leader-based and follows the standard PBFT framework~\cite{castro1999practical}.
As mentioned, it directly solves $2$-round single-shot BFT, and an extension to BFT SMR can be found in our complementary paper~\cite{abraham2021fast}.
Our protocol is also optimal in terms of resilience, as we can show that any \psyncbbshort protocol with $n\leq 5f-2$ will have a \latency of at least $3$ rounds (Theorem~\ref{thm:psync:lb:3}).

\textbf{Relation to the previous work~\cite{martin2006fast}.}
\label{sec:previouswork}
The authors of FaB~\cite{martin2006fast} propose a $2$-round PBFT with $n\geq 5f+1$ and claim the resilience is optimal by proving a lower bound that any Byzantine agreement protocol with $n\leq 5f$ cannot always commit within $2$ round.
However, their lower bound assumes a family of Paxos-like protocol that separates proposers from acceptors. In our protocol, parties act as both proposers and acceptors, and therefore we are able to circumvent the lower bound and improve the resilience to $n\geq 5f-1$.

\textbf{Intuition.}
Before presenting our protocol, it is helpful to briefly explain how FaB achieves $2$-round commit with $n=5f+1$ parties. 
In FaB, the \latency of $2$ rounds consists $1$ round of proposing and $1$ round of voting, thus reducing $1$ round of voting compared with PBFT~\cite{castro1999practical}. A value $v$ is safe to be committed if it is voted by $n-f=4f+1$ parties, among which at least $3f+1$ must be honest. Then, during the view-change, any set of $4f+1$ view-change messages must contain at least $(4f+1)+(3f+1)-n=2f+1$ messages from those honest parties that voted for $v$. Since $2f+1$ is a majority of $4f+1$, the next leader can re-propose the majority value to ensure safety across different views. If we reduce the number of parties, i.e., $n=5f$, then the set of $4f$ view-change messages may contain two disjoint sets of $2f$ messages supporting two different values respectively, and the next leader is unable to break the tie.

\textsl{Main observation: detecting leader equivocation with authentication.}
Our protocol has \latency of $2$ rounds, consisting $1$ round of proposing and $1$ round of voting.
The main observation is that parties can {\em detect the malicious behavior of the leader with authentication} and thus further reduce the number of parties to $n=5f-1$. 
More specifically, leader equivocation can be detected by honest parties when they receive more than one value signed by the leader. Then, {\em if any honest party detects that the leader is Byzantine, it can wait for one more view-change message from parties other than the broadcaster}. Therefore, the set of $n-f=4f-1$ view-change messages under leader equivocation contains at most $f-1$ messages from the Byzantine parties and thus at least $3f$ messages from the honest parties.
When any honest party commits $v$ by receiving $n-f=4f-1$ votes for $v$, at least $3f-1$ honest parties have voted for $v$. During view-change, honest party either receives $4f-1$ view-change messages containing $\geq (4f-1)-f-f=2f-1$ messages for $v$ and no message for other values, or detect leader equivocation. 
For the latter case, the set of view-change messages must contain at least $3f+(3f-1)-(4f-1)=2f$ messages from the honest parties who voted for $v$, which is the majority and any honest party can thus lock on $v$ during view change.

\begin{figure}[t]
    \centering
    \begin{mybox}
    
    The $(5f-1)$-\vbbshort protocol uses the following rule for certificate check. $L_w$ is the leader for view $w$.

{\bf Certificate Check.} 
    $\fc$ is a valid certificate of view $w$ iff it contains $\geq 4f-1$ signed messages from distinct parties, where each message from party $j$ is either 
    $\langle \bot, w \rangle_{j}$,
    or 
    $\langle v, w \rangle_{L_w, j}$ where $\ff(v)=true$.
    Moreover,
    $\fc$ {\em locks} a value $v\neq\bot$ iff
    (1) 
    it contains $\geq 2f-1$ $\langle v, w \rangle_{L_w, j}$ for any $j\in[n]$, and no $\langle v', w \rangle_{L_w, j}$ for any $v'\neq v, j\in[n]$
    , or 
    (2) 
    it contains $\geq 2f$ $\langle v, w \rangle_{L_w,j}$ where $j\neq L_w$.

    For initialization, $\emptyset$ is also a valid certificate of view $0$, locking any externally valid value $v\neq\bot$.
    Note that by definition, if $\fc$ locks $v\neq\bot$, it does not lock on any $v'\neq v$.
    $\fc$ ranks higher with a higher view number.

\end{mybox}
    \caption{Certificate Check for $(5f-1)$-\vbbshort}
    \vspace{-1em}
    \label{fig:certcheck}
\end{figure}

\begin{figure}[t!]
    \centering
    \begin{mybox}
    
    The protocol proceeds in view $w=1,2,...$ each with a leader $L_{w}$, and the first leader $L_1$ being the designated broadcaster.
    Each party keeps the highest certificate $\fc_h$ received, initialized as $\emptyset$.
    The parties will ignore any message for value $v\neq \bot$ that is not externally valid.
    
\begin{enumerate}
    
    \item\label{pbft:step:propose} \textbf{Propose.} 
    The leader $L_w$ sends $\langle \texttt{propose}, \langle v, w\rangle_{L_w},\fs \rangle_{L_w}$ to all parties.
    If $w=1$, then $v$ is the input of the leader and $\fs=\bot$; otherwise $v,\fs$ are specified in the {\em Status} step.
    
    \item\label{pbft:step:vote} \textbf{Vote.} 
    Upon receiving the first proposal in the form of $\langle \texttt{propose}, \langle v, w\rangle_{L_w}, \fs \rangle_{L_w}$ from the leader $L_w$,
    if 
    \begin{itemize}[itemsep=0pt,topsep=0pt]
        \item $w=1$, or
        \item $\fs$ is a valid certificate of view $w-1$ that locks $v$, or
        \item $\fs$ contains $4f-1$ valid \status messages of view $w-1$ each with a valid certificate of view $\leq w-1$ that locks some non-$\bot$ value, and the highest certificate in $\fs$ locks $v$,
    \end{itemize}
    multicast a \vote message in the form of $\langle \vote, \langle v, w\rangle_{L_w,i} \rangle_i$.

    \item\label{pbft:step:commit} \textbf{Commit.}
    When receiving $4f-1$ distinct signed \vote messages of view $w$ for the same value $v$, forward these $4f-1$ \vote messages to all other parties, and commit $v$.
    
    \item\label{pbft:step:timeout} \textbf{Timeout.}
    If not committed within $4\Delta$ after entering view $w$, timeout view $w$, which means
    (1) if timeout before voting, stop voting for view $w$ and multicast $\langle \timeout, \langle \bot ,w\rangle_{i} \rangle_i$;
    (2) if timeout after voting for value $v$, multicast $\langle \timeout, \langle v,w\rangle_{L_w,i} \rangle_{i}$.
    
    \item\label{pbft:step:newview} \textbf{New View.} Upon receiving $4f-1$ valid \timeout messages of view $w-1$ from distinct parties that contains only one non-$\bot$ value signed by $L_{w-1}$ and the party, or receiving $4f-1$ valid \timeout messages from parties other than $L_{w-1}$, perform the following.
    
    Forward these \timeout messages, update $\fc_h$ if $4f-1$ signatures form a valid certificate of view $w-1$ that locks any $v\neq \bot$, timeout view $w-1$ if haven't, and enter view $w$.
    Send a \status message in the form of $\langle \status, w-1, \fc_h \rangle_i $ to the leader $L_w$.
    
    \item\label{pbft:step:status} \textbf{Status.}
    After entering view $w$ and receiving $4f-1$ valid \status messages of view $w-1$ each with a valid certificate $\fc$ of view $\leq w-1$ that locks some non-$\bot$ value,
    the leader $L_w$ sets a proof $\fs$ and a proposal value $v$ as follows:
    \begin{itemize}[itemsep=0pt,topsep=0pt]
        
        \item If the \status messages contain a valid certificate $\fc$ of view $w-1$, set $\fs=\fc$ and the proposal to the value $v$ that $\fc$ locks.
        
        \item Otherwise, set $\fs$ to be the set of $4f-1$ valid \status messages of view $w-1$ received, and set proposal to be the value $v$ that the highest valid certificate in $\fs$ locks.
    \end{itemize}
\end{enumerate}
    \end{mybox}
    \caption{$(5f-1)$-\vbbshort Protocol with \latency of $2$ rounds}
    \vspace{-1em}
    \label{fig:psyncbb}
\end{figure}

\textbf{Certificate check.}
The observation above partly explains the intuition of the {\em certificate check} in Figure~\ref{fig:certcheck}, which defines a valid certificate that {\em locks} a value for view-change. 
When a value $v$ is committed at any honest party, to ensure agreement, we want all honest parties to lock $v$ after receiving a valid certificate and only vote for $v$ in any future views.
A valid certificate $\fc$ of view $w$ consists of at least $4f-1$ signed tuple, consisting a value (can be $\bot$) and a view number $w$, from different parties. For brevity, we will just call them signed values. The signed values are from the \vote messages (Step~\ref{pbft:step:vote}) or \timeout messages (Step~\ref{pbft:step:timeout}) of the $(5f-1)$-\vbbshort protocol.
A valid certificate $\fc$ {\em locks} a value $v\neq \bot$ when there may be some honest party that already commits $v$, and thus any honest party that receives such $\fc$ should lock on value $v$ for agreement.
Similar to the earlier argument, if any value $v$ is committed, there must be $\geq 3f-1$ honest parties that voted for $v$ and signed $v$. If $\fc$ contains only one value signed by the leader $L_w$ of view $w$, then $\fc$ must contain at least $2f-1$ signed $\langle v,w \rangle$.
Otherwise, if $\fc$ contains more than one value signed by the leader $L_w$, then $L_w$ is Byzantine as it equivocated, and $\fc$ should contain $4f-1$ signed values from parties other than $L_w$, and thus at least $2f$ signed $\langle v,w \rangle$, which is the majority.
Therefore, any honest party will lock on $v$ when receiving a valid certificate $\fc$ that locks $v$, since $v$ may have been committed by some honest parties.

\textbf{Protocol description.}
The protocol proceeds in views starting from $1$, each view has a designated leader (by round-robin for instance), and the first leader is the designated broadcaster.
Each party locally keeps the highest certificate $\fc$ it has ever seen, which is $\emptyset$ initially.
Except for view $1$, each party enters a new view $w$ in Step~\ref{pbft:step:newview} by gathering $4f-1$ \timeout messages of view $w-1$ (see Step~\ref{pbft:step:timeout}) satisfying the conditions described in Step~\ref{pbft:step:newview}. Here the purpose of gathering \timeout messages is to ensure that all honest parties will lock on the value $v$ if any honest party has committed $v$.
Thus, if detect equivocation of the previous leader, the party will wait for one more \timeout message from other parties. 
Then, if the received $4f-1$ \timeout messages form a valid certificate that locks any value $v\neq \bot$, the party updates its highest certificate $\fc$. The party also timeouts the old view (as defined in Step~\ref{pbft:step:timeout}), enters the new view and sends a \status message with $\fc$.
The broadcaster, who is also the leader of view $1$, can just propose its input as the proposal.
For any other leader of view $\geq 2$, it can propose a value and a proof based on the set of $4f-1$ \status messages received.
If its highest certificate $\fc$ is updated in Step~\ref{pbft:step:newview}, meaning possibly some honest party has committed a value in the previous view, the leader proposes the same value and its $\fc$ to ensure agreement.
Otherwise, the leader sets the proposal to be the value locked by the highest certificate among the \status messages received, and attaches all these \status messages as a proof.
After receiving the proposal, each party will check if the proposed value and the proof are produced according to the above steps, and multicast a \vote message for the value in Step~\ref{pbft:step:vote} if the check passes. 
When receiving $4f-1$ votes for the same value $v$, the party forwards these votes and commit $v$.
Otherwise, if the party does not commit within $4\Delta$ time after entering view $w$, which means either the network is bad or the leader is Byzantine, it will timeout the current view by sending a \timeout message. If the party timeouts after voting for some value $v$, it sends \timeout message with the value $v$, otherwise it sends \timeout message with $\bot$.
When enough \timeout messages are collected, the party enters the new view $w+1$, as in Step~\ref{pbft:step:newview}.

\textbf{Proof of Correctness.} The correctness proof can be found in Appendix~\ref{sec:psyncbb:proof}.

\subsection{$3f+1\leq n\leq 5f-2$,  Matching Lower and Upper Bounds of $3$ Rounds}
\label{sec:psync:2}

Recall that \psyncbbshort solves a single-shot of BFT SMR under partial synchrony, hence \psyncbbshort is solvable if and only if $n\geq 3f+1$ as BFT SMR is solvable if and only if $n\geq 3f+1$ under partial synchrony~\cite{dwork1988consensus}.
\daniel{double check}
For the remaining case of $ 3f+1\leq n\leq 5f-2$, we show a lower bound of $3$ rounds on its \latency using standard indistinguishability arguments. 
This is tight given the PBFT protocol~\cite{castro1999practical} solves \vbbshort with \latency of $3$ rounds and $n\geq 3f+1$.

\begin{theorem}\label{thm:psync:lb:3}
    Any authenticated \psyncbblong that is resilient to $f\geq (n+2)/5$ faults must have a \latency of at least $3$ rounds under partial synchrony.
\end{theorem}


\begin{proof}

The proof is illustrated in Figure~\ref{fig:psyncbb-lb}.
Suppose there exists a \psyncbbshort protocol $\Pi$ that has \latency of $2$ rounds under $n\leq 5f-2$. 
By definition, when the broadcaster is honest and the network is synchronous $(GST=0)$, $\Pi$ ensures all honest parties commit after delivering all \clone and \cltwo messages, even if any \clthree message is not delivered yet. 
Divide $n\leq 5f-2$ parties into one broadcaster $s$, and five disjoint groups $A,B,C,D,E$ where $A,C,D$ have size $\leq f-1$ and $B,E$ have size $f$.
For brevity, we will often use a group to refer to all the parties in that group.
We will construct the following $5$ executions.
In all constructed executions, all messages are delivered by the recipient after $\Delta$ time by default, and we will explicitly specify the messages that are delayed by the adversary due to asynchrony.
Also, we focus on the messages between different groups, and assume by default the party sends and delivers any message within its group as well.
\begin{itemize}[leftmargin=*]
    \item Execution  1. The network is synchronous ($GST=0$).
    The broadcaster is honest and proposes $0$ to all parties. The $f$ parties in $E$ are Byzantine that only send \clone messages faithfully to $B,C,D$ and no message to $A$.
    Since the broadcaster is honest, by the \latency guarantee and validity, $A,B,C,D$ commit $0$ within $2$ rounds after delivering all \clone and \cltwo messages.
    
    \item Execution  5. Symmetric to Execution 1, the network is synchronous ($GST=0$), the broadcaster is honest and proposes $1$ to all parties. The $f$ parties in $B$ are Byzantine that only send \clone messages faithfully to $C,D,E$ and no message to $A$.
    By the \latency guarantee and validity, $A,C,D,E$ commit $1$ within $2$ rounds after delivering all \clone and \cltwo messages.
    
    \item Execution  3. The network is synchronous ($GST=0$).
    The broadcaster is Byzantine; it sends $0$ to $B,C$, and $1$ to $D,E$. 
    Then the broadcaster behaves to $B,C$ the same way as the broadcaster to $B,C$ in Execution 1, and behaves to $D,E$ the same way as the broadcaster to $D,E$ in Execution 5.
    The $f-1$ parties in $A$ are Byzantine and only send \clone messages faithfully to $B,C,D,E$.
    By termination, $B,C,D,E$ eventually commit some value.

    \item Execution  2. The network is asynchronous before $B,E$ commit ($GST$ comes after $B,E$ commit).
    The broadcaster is Byzantine; it sends $0$ to $A,B,C$, and $1$ to $E$. Then the broadcaster behaves to $A,B,C$ the same way as the broadcaster to $A,B,C$ in Execution 1, and behaves to $E$ the same way as the broadcaster to $E$ in Execution 5.
    The $f-1$ parties in $D$ are Byzantine; they behave to $A$ the same way as $D$ to $A$ in Execution 1, and behave to the others the same way as $D$ to the others in Execution 3.
    Any message other than \clone messages from $A$ to the rest of the parties is delayed indefinitely until $GST$. 
    Any message from $E$ to $A$ is delayed indefinitely until $GST$.
    
    \item Execution  4. 
    Symmetric to Execution 2, the network is asynchronous before $B,E$ commit ($GST$ comes after $B,E$ commit); the broadcaster is Byzantine; it sends $0$ to $B$, and $1$ to $A,D,E$. 
    Then the broadcaster behaves to $A,D,E$ the same way as the broadcaster to $A,D,E$ in Execution 5, and behaves to $B$ the same way as the broadcaster to $B$ in Execution 1.
    The $f-1$ parties in $C$ are Byzantine; they behave to $A$ the same way as $C$ to $A$ in Execution 5, and to the others the same way as $C$ to the others in Execution 3.
    Any message other than \clone messages from $A$ to rest of the parties is delayed indefinitely until $GST$.
    Any message from $B$ to $A$ is delayed indefinitely until $GST$.
\end{itemize}

\begin{figure*}[tb]
    \centering
    \includegraphics[width=\textwidth]{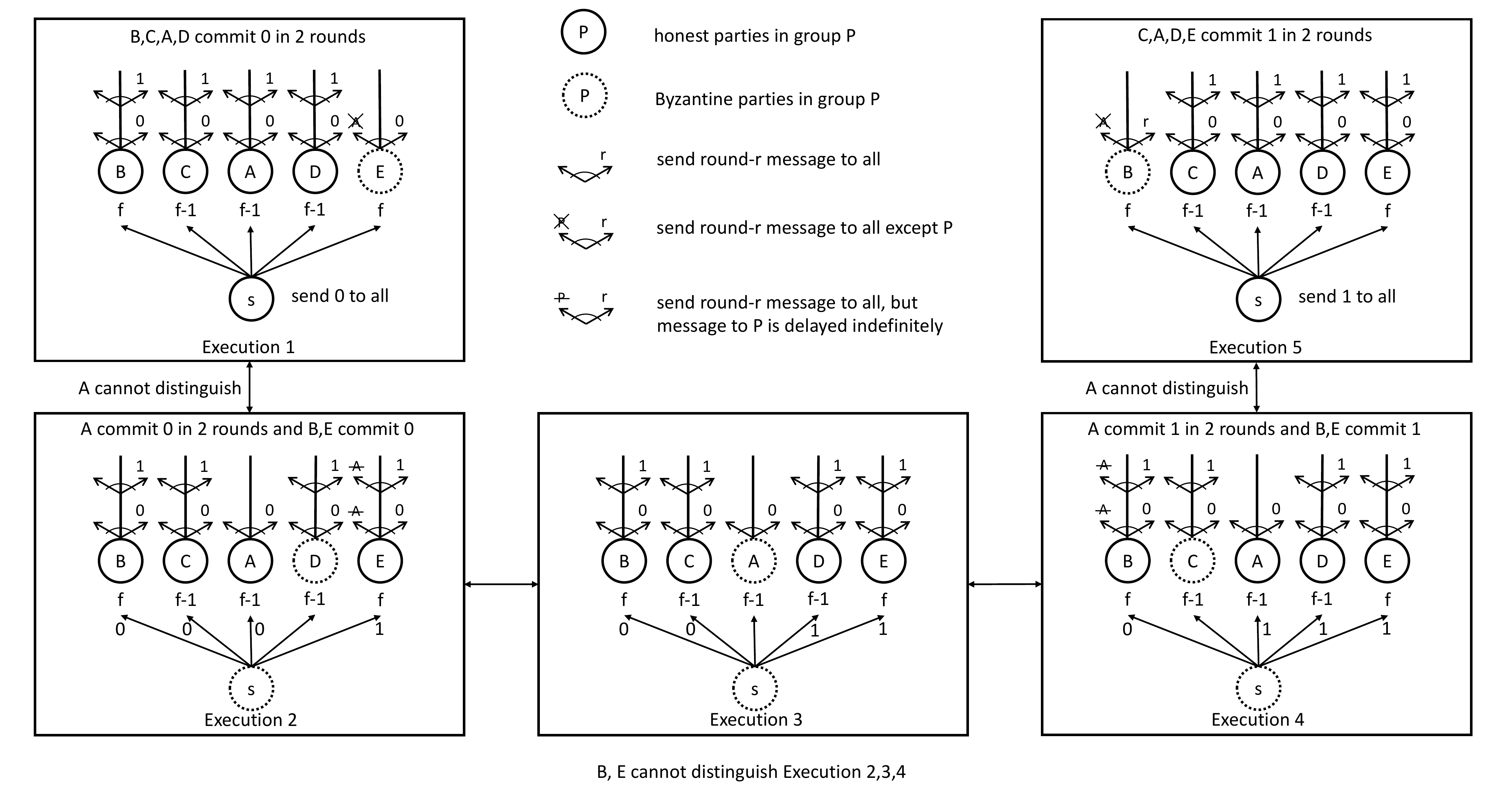}
    \caption{\psyncbbshort Lower Bound: $3$ rounds with $n=5f-2$. 
    Dotted circles denote Byzantine parties. In Execution 1 (5), the network is synchronous ($GST=0$), $E$ ($B$) are Byzantine and only send \clone messages faithfully to rest of the parties except $A$. 
    In Execution 3, the network is synchronous ($GST=0$), $A$ are Byzantine and only send \clone messages faithfully to rest of the parties.
    In Execution 2 (4), the network is asynchronous before $B,E$ commits, messages from $A$ after round 0 are delayed indefinitely until $GST$, and messages from $E$ ($B$) to $A$ are delayed indefinitely until $GST$.
    }
    \label{fig:psyncbb-lb}
    \vspace{-1em}
\end{figure*}

We show the following indistinguishability and contradiction. 
\begin{itemize}[leftmargin=*]
    \item $A$ cannot distinguish Execution 1 and 2 after delivering all \clone and \cltwo messages but before delivering any \clthree message. 
    \begin{itemize}
        \item  First we show that any \clone or \cltwo message from $B,C$ is identical in both executions. 
        Any \clone message only depends on the party's initial state, thus is identical since $B,C$ are honest.
        Any \cltwo message depends on the \clone messages the party delivers. In the two executions, the broadcaster sends the same value of $0$ to $B,C$, and $A,D,E$ all follow the protocol to send the same \clone messages to $B,C$, so $B,C$ send the same \cltwo messages to $A$.
        \item For $D$ and the broadcaster, in Execution 2, the Byzantine parties in $D$ and the broadcaster behave identically to $A$ as in Execution 1. For $E$, no message from $E$ is delivered by $A$ in both executions. 
    \end{itemize}
    Since $A$ is honest and delivers identical \clone and \cltwo messages in Execution 1 and 2, $A$ cannot distinguish Execution 1 and 2 before delivering any \clthree message. 
    Since $A$ commits $0$ in Execution 1, $A$ also commits $0$ in Execution 2. Then by agreement and termination, $B,C,E$ also commit $0$ in Execution 2.
    
    \item Similar to the argument above, 
    $A$ cannot distinguish Execution 4 and 5 after delivering all \clone and \cltwo messages but before delivering any \clthree message. 
    Therefore, $A,B,D,E$ also commit $1$ in Execution 4.

    \item $B,E$ cannot distinguish Execution 2 and 3.
    For $A$, in both executions any party not in $A$ only delivers the same \clone messages from $A$. For $D$ and broadcaster, in Execution 2, they behave identically as in Execution 3. Since $B,C,E$ are honest, they will behave identically in both executions.
    Therefore, $B,E$ cannot distinguish Execution 2 and 3, and will eventually commit both executions to satisfy termination. 
    Since $A$ commit $0$ in Execution 2, $B,E$ will eventually commit $0$ in Execution 2 by agreement, and thus commit $0$ in Execution 3 as well.
    
    \item Similar to the argument above, $B,E$ cannot distinguish Execution 4 and 3, and will eventually commit $1$ in Execution 3.
\end{itemize}
We proved that $B,E$ commit both $0$ and $1$ in Execution 3, contradiction. Hence such a protocol $\Pi$ cannot exist. \qedhere
\end{proof}

\section{Synchronous Byzantine Fault-tolerant Broadcast}

In this section, we present lower and upper bound results for broadcast under synchrony.
To strengthen the results, we prove all the lower bounds for Byzantine reliable broadcast (thus also apply to Byzantine broadcast), and all the upper bounds for Byzantine broadcast (thus also apply to Byzantine reliable broadcast).
Furthermore, all lower bound results assume synchronized start, and all upper bound results assume unsynchronized start, except for the case of $n/3<f<n/2$ where the tight bounds depend on the assumption.

The situations of $f<n/3$ and $f=n/3$ are relatively easy and often modified from known results. Hence, we deferred their details to appendices. 
We will instead focus on the harder and more surprising case of $n/3 < f < n/2$.

\subsection{$0<f<n/3$, Matching Lower and Upper Bounds of $2\delta$}
\label{sec:sync:1}

\begin{theorem}\label{thm:lb:sync:2d}
    Any Byzantine reliable broadcast protocol that is resilient to $f> 0$ faults must have a \latency at least $2\delta$, even with synchronized start.
\end{theorem}

For the lower bound, early-stopping result~\cite{dolev1990early} implies any BB protocol must have \latency of $2\delta$, and here we slightly strengthen the result for BRB.
Intuitively, 
it takes at least $\delta$ to receive from the broadcaster and another $\delta$ for parties to exchange the proposal received from the broadcaster, otherwise parties may commit in different values under a Byzantine broadcaster if the \latency is $<2\delta$.

For a matching upper bound, the protocol relies on the standard quorum intersection technique. The parties vote for the first proposal they received, and commit within $2\delta$ time if they receive $n-f$ votes on the same value. By quorum intersection, there cannot be $n-f$ votes on different values. Otherwise, a BA is used to guarantee agreement when the broadcaster is Byzantine.
Details can be found in Appendix~\ref{sec:sync<1/3}.

\subsection{$f=n/3$, Matching Lower and Upper Bounds of $\Delta+\delta$}
\label{sec:sync:2}

\paragraph{BRB lower bound $\Delta+\delta$ under synchronized start and $f\geq n/3$.}
\label{sec:sync:lb:syncstart}
Based on the $\Delta$ lower bound proof for BB in~\cite{synchotstuff}, we show a more accurate lower bound of $\Delta+\delta$ on the \latency for BRB using similar proof techniques. 
Intuitively, since $f\geq n/3$, the adversary can split the honest parties into two disjoint groups, each with a different proposed value from the Byzantine broadcaster. 
If a protocol can commit before $\Delta+\delta$, the two groups will commit conflicting values before $\Delta+\delta$, which is before they communicate any information about the broadcaster's proposed value. Details can be found in Appendix~\ref{sec:sync:lb:syncstart}.

\begin{theorem}\label{thm:lb:sync:D+d}
    Any Byzantine reliable broadcast protocol that is resilient to $f\geq n/3$ faults must have a \latency at least $\Delta+\delta$, even with synchronized start. 
\end{theorem}

\paragraph{BB upper bound $\Delta+\delta$ under unsynchronized start and $f=n/3$.}
To show that $\Delta+\delta$ is the tight \latency bound for the case of $f=n/3$, we show such a BB protocol in Figure~\ref{fig:sync:3f}, with correctness proof in Appendix~\ref{sec:sync=1/3}.

Each party starts the protocol at most $\clockskew$ time apart with a local clock starting at $0$, and it is guaranteed that $\clockskew\leq \delta$ by any clock synchronization protocol~\cite{dolev1995dynamic, abraham2019synchronous}.
Since the value of $\delta$ is unknown to the protocol designer or any party, all parties set the parameter $\clockskew=\Delta$ in the protocol.
The broadcaster first multicasts its proposed value, and any party that receives the first valid proposal will vote for the proposal. Meanwhile, the party also starts a timer to wait for $\Delta$ time for equivocation detection. 
Recall that we say a party detects equivocation if it receives messages containing different values signed by the broadcaster.
If the party detects no different value signed by the broadcaster during the above $\Delta$ waiting period, it may commit earlier as follows.
When the party receives $n-f$ votes for the same value $v$, it forwards these votes, and if the local time is $\leq 2\Delta+\clockskew$ when receiving the votes, the party commits and locks the value, and multicasts a commit message.
At time $3\Delta+2\clockskew$, each party checks the set of votes it received. If the party receives $n-f$ votes for one value, it locks that value. Otherwise, if there exist two sets of $n-f$ votes on different values, the parties in the intersection of the two sets must be all Byzantine since they voted for two values. Hence, an honest party can identify all the Byzantine parties and commit and lock the same value from honest parties.
Then, all the parties participate in an instance of BA with input $\lock$, and commit the output if they haven't committed.

Observe that when $f=n/3$, all $f$ Byzantine parties will expose themselves if they try to double vote and make different honest parties lock on different values.
As we will show next in Section~\ref{sec:sync:syncstart} and \ref{sec:sync:unsyncstart}, if the number of faults exceeds $n/3$, the tight bound of \latency becomes dependent on the clock synchronization assumption.  

\begin{figure}[tb]
    \centering
    \begin{mybox}
    Initially, every party $i$ starts the protocol at most $\delta$ time apart with a local clock and sets $\lock=\bot$, $\clockskew=\Delta$.
\begin{enumerate}
    \item\label{bb2:step:propose} \textbf{Propose.} The designated broadcaster $L$  with input $v$ sends $\langle \texttt{propose}, v\rangle_L$ to all  parties.

    \item\label{bb2:step:vote} \textbf{Vote.} 
    When receiving the first valid proposal $\langle \texttt{propose}, v\rangle_L$ from the broadcaster, 
    send a vote to all parties in the form of $\langle \texttt{vote}, \langle \texttt{propose}, v\rangle_L \rangle_i$ where $v$ is the value of the proposal. Set $\texttt{vote-timer}$ to $\Delta$ and start counting down.
    
    \item\label{bb2:step:commit} \textbf{Commit.}
    When $\texttt{vote-timer}$ reaches $0$,
    if the party detects no equivocation, it does the following: When receiving $n-f$ signed votes for $v$, forward these $n-f$ votes to all other parties. 
    If the $n-f$ votes for $v$ are received before local time $2\Delta+\clockskew$, commit $v$, set $\lock=v$ and send $\langle \texttt{commit}, v \rangle_i$ to all parties.
    
    \item\label{bb2:step:ba} \textbf{Lock and Byzantine agreement.}
    At local time $3\Delta+2\clockskew$, if the party receives $n-f$ signed votes for only one value $v$, it sets $\lock=v$. 
    Otherwise, if two sets of $n-f$ signed votes for different values are received, let $\ff$ be the set of parties in the intersection of these two vote sets.
    If the party receives $\langle \texttt{commit}, v \rangle_j$ from any party that is not in $\ff$, it commits $v$ 
    and set $\lock=v$.
    
    Then, invoke an instance of Byzantine agreement with $\lock$ as the input. 
    If not committed, commit on the output of the Byzantine agreement. 
    Terminate.
\end{enumerate}
    \end{mybox}
    \caption{$(\Delta+\delta)$-$n/3$-BB Protocol with $f\leq n/3$}
    \label{fig:sync:3f}
    \vspace{-1em}
\end{figure}

\subsection{$n/3<f<n/2$ and Synchronized Start, Matching Lower and Upper Bounds of $\Delta+\delta$}
\label{sec:sync:syncstart}

For $n/3<f<n/2$, the tight bound of \latency depends on the assumption on clock synchronization.
If all honest parties start the protocol at the same time and have synchronized clocks ($\clockskew=0$), then the tight bound is $\Delta+\delta$ by Theorem~\ref{thm:lb:sync:D+d} from the previous section and Theorem \ref{thm:ub:sync:syncstart} in this section.
Otherwise, if there exists a clock skew of $\clockskew\geq 0.5\delta$, the tight bound becomes $\Delta+1.5\delta$ by Theorem~\ref{thm:lb:sync:D+1.5d} and \ref{thm:ub:sync:unsyncstart} later in Section~\ref{sec:sync:unsyncstart}.

\paragraph{BRB lower bound $\Delta+\delta$ under synchronized start and $f \geq n/3$.}
See Theorem~\ref{thm:lb:sync:D+d}.

\paragraph{BB upper bound $\Delta+\delta$ under synchronized start and $f<n/2$.}
Now we present a protocol $(\Delta+\delta)$-BB that works under $f<n/2$ and synchronized start, and has a optimal \latency of $\Delta+\delta$.
The protocol is presented in Figure~\ref{fig:sync:D+d}.
Every party locally sets its $\lock$ to be some default value $\bot$ and its $\rank=\Delta+1$, and starts the protocol simultaneously at time $0$.
The broadcaster first multicasts its proposal, and any party that receives the first valid proposal at time $d\leq \Delta$ will multicast a \vote for the proposal containing the time $d$.
When receiving $f+1$ votes for the same value $v$, if there exists a $t\in[0,\Delta]$ such that the party detects no equivocation  within time $t+\Delta$ and all $f+1$ votes contain time $d\leq t$, the party can commit $v$ and forward these votes.
Otherwise, the party updates its lock if receiving $f+1$ votes of higher rank.
Finally, all parties participate in an instance of BA at time $4\Delta$ with input $\lock$, and commit the output of the BA if have not committed.
The correctness proof can be found in Appendix~\ref{sec:ubproof:1+1}.

\begin{figure}[tb]
    \centering
    \begin{mybox}
    Initially, every party $i$ starts the protocol at the same time from $0$, sets $\lock=\bot$ and $\rank=\Delta+1$.
\begin{enumerate}
    \item\label{bb4:step:propose} \textbf{Propose.} The designated broadcaster $L$  with input $v$ sends $\langle \propose, v\rangle_L$ to all  parties.
    
    \item\label{bb4:step:vote} \textbf{Vote.}
    When receiving the first valid proposal $\langle \propose, v\rangle_L$ from the broadcaster at time $d$ where $d\leq \Delta$,  send a \vote in the form of $\langle \vote, d, \langle \propose, v\rangle_L \rangle_i$ to all parties.

    \item\label{bb4:step:commit} \textbf{Commit and Lock.}
    For any $t\in [0,\Delta]$,
    if party $i$ detects no equivocation within time $t+\Delta$, and receives $f+1$ signed \vote messages for the same value $v$ each with $d\leq t$, party $i$ commits $v$ and forward these $f+1$ \vote messages. 

    For any $t\in [0,\Delta]$,
    within time $2\Delta+t$, if receive $f+1$ signed \vote messages for the same value $v$ each with $d\leq t$, 
    and $\rank>t$, update $\lock=v$ and $\rank=t$.

    \item\label{bb4:step:ba} \textbf{Byzantine agreement.}
    At time $4\Delta$, invoke an instance of Byzantine agreement with $\lock$ as the input. 
    If not committed, commit on the output of the Byzantine agreement. 
    Terminate.

\end{enumerate}
    \end{mybox}
    \caption{$(\Delta+\delta)$-BB Protocol with $f<n/2$ and synchronized start}
    \vspace{-1em}
    \label{fig:sync:D+d}
\end{figure}

\subsection{$n/3<f<n/2$ and Unsynchronized Start, Matching Lower and Upper Bounds of $\Delta+1.5\delta$}
\label{sec:sync:unsyncstart}

Interestingly, when the clocks at each party are not perfectly synchronized and the parties therefore do not start the protocol at the same time, the tight bound for \latency increases when the number of faults is $n/3<f<n/2$.
The tight bound for \latency is $\Delta+1.5\delta$ under unsynchronized start, which consists of a $1.5\delta$ term which is not an integer multiple of $\delta$.
This is perhaps the most interesting and surprising result of this paper, as it involves a protocol that is very different from conventional ones whose latency have always been an integer multiple of the message delay.

\paragraph{BRB lower bound $\Delta+1.5\delta$ under unsynchronized start and $ f>n/3$.}
We first present the lower bound result that shows no BRB protocol can have \latency less than $\Delta+1.5\delta$ under unsynchronized start and $f>n/3$ faults.
The lower bound proof is based on the standard indistinguishablility argument, where the proof constructs multiple executions that are indistinguishable to certain honest parties, to derive a safety violation in some execution.

\begin{theorem}\label{thm:lb:sync:D+1.5d}
    Any Byzantine reliable broadcast protocol with unsynchronized start and is resilient to $f>n/3$ faults must have a \latency at least $\Delta+1.5\delta$.
\end{theorem}

\begin{figure*}[tb]
    \centering
    \includegraphics[width=0.8\textwidth]{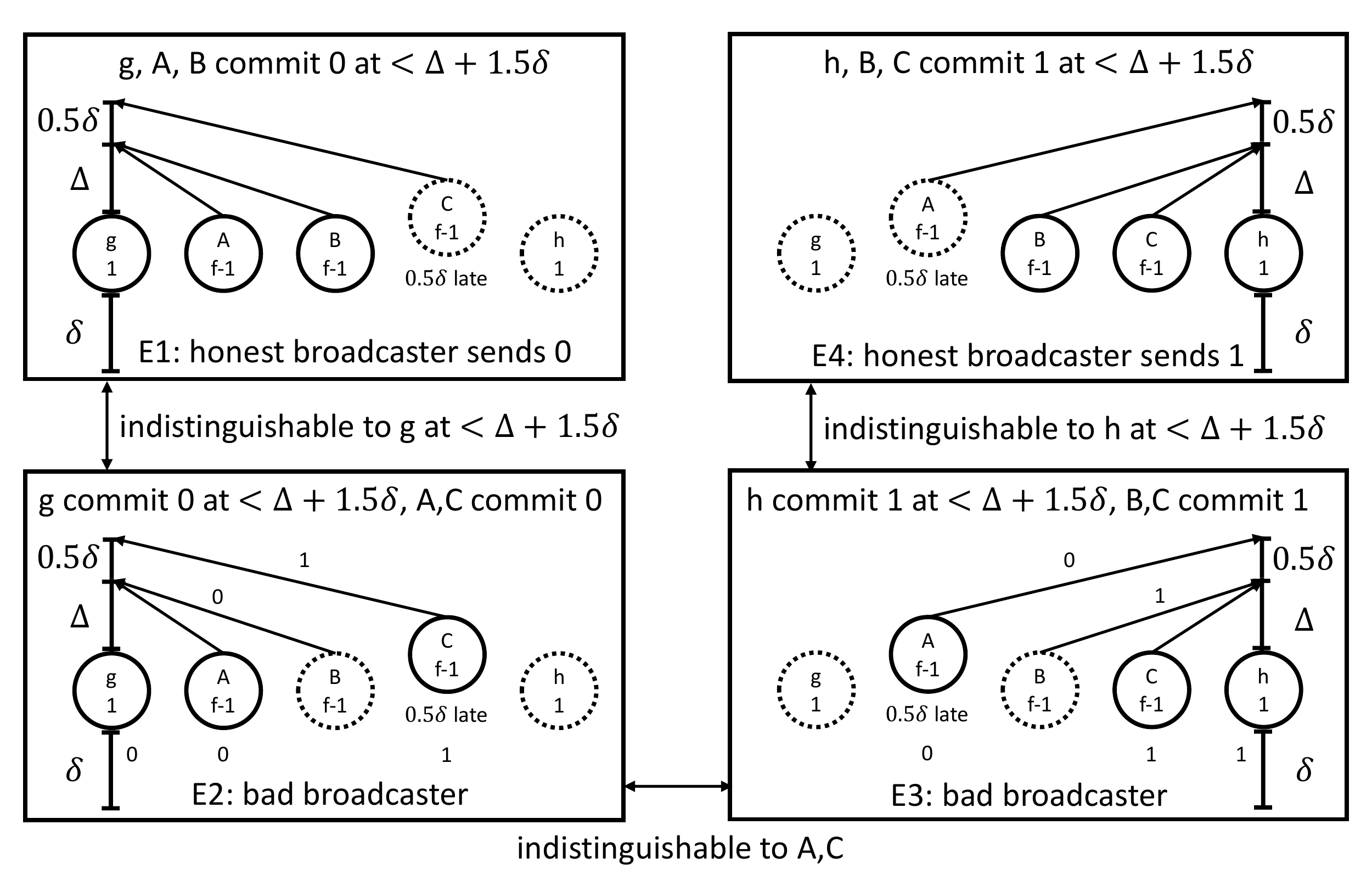}
    \caption{Illustration of the proof sketch for the $\Delta+1.5\delta$ lower bound.}
    \vspace{-1em}
    \label{fig:BB-lb-intuition}
\end{figure*}

\begin{proof}[Proof Sketch]
    Illustrated in Figure~\ref{fig:BB-lb-intuition}, and the complete proof can be found in Appendix~\ref{sec:lbproof:1+1.5}.
    Suppose there exists a BRB protocol that has \latency $<\Delta+1.5\delta$ under $n<3f$.
    As mentioned in Section~\ref{sec:prelim}, we assume the clock skew $\clockskew= 0.5\delta$ due to the lower bound for clock skew~\cite{attiya2004distributed}.
    We can construct $4$ executions as follows, and use the standard indistinguishability argument to derive contradictions.
    We divide the parties into $5$ groups $A,B,C$ of size $\leq f-1$ each, and $g,h$ of size $1$ each. The Byzantine parties are denoted by the dotted circles in the figure. The broadcaster is some party in $B$.
    
    In execution E1 with actual message delay bound $\delta$, the broadcaster is honest and sends $0$ at global time $0$, which is received by all parties at global time $\delta$. By assumption, honest parties in $g,A,B$ commit $0$ before global time $\Delta+1.5\delta$. Byzantine parties in $C$ pretend that it starts the protocol $0.5\delta$ time later.

    In execution E2 with actual message delay bound $\Delta$, the broadcaster is Byzantine, and sends $0$ to $g,A$ and $1$ to $C$.
    The honest parties in $C$ start the protocol $0.5\delta$ time later due to the clock skew, and receive $1$ from the broadcaster at local time $\delta$ and global time $1.5\delta$. Moreover, any message from $C$ to $g$ has delay $\Delta$. Therefore, {\em before global time $\Delta+1.5\delta$, $g$ cannot learn the fact that $C$ receive $1$ from the broadcaster}. We can carefully construct the executions such that $g$ cannot distinguish E1 and E2 before global time $\Delta+1.5\delta$, and will commit $0$ in E2 as well.
    
    Executions E3 and E4 are the symmetric case of executions E2 and E1, respectively.
    In E4, broadcaster is honest and sends $1$, and $h,B,C$ commits $1$ before global time $\Delta+1.5\delta$.
    In E3, broadcaster is Byzantine, but $h$ cannot learn that before commit since the equivocating message from $A$ reaches $h$ at global time $\Delta+1.5\delta$, as $A$ starts the protocol $0.5\delta$ time later due to clock skew.
    Similarly, $h$ cannot distinguish E3 and E4 before global time $\Delta+1.5\delta$, and will commit $1$ in E3 as well.
    
    The last step to complete the proof is to show that $A,C$ cannot distinguish E2 and E3.
    The intuition is that, with the $f-1$ Byzantine parties in $B$ equivocating to $A$ and $C$, the honest parties in $A,C$ cannot decide between $g$ and $h$ which is the honest party that actually commits.
    Also, the message delays between $A,C$, and all pairs of honest parties are controlled by the adversary, such that the honest parties in $A,C$ cannot tell which of them start the protocol $0.5\delta$ time later than the rest. For instance, the adversary can make the delay from $C$ to $A$ to be $\Delta-\delta$ in E2 and $\Delta$ in E3. Then, the differences in the delays compensate the differences in when $A,C$ start their protocol.
    Once we proved that $A,C$ cannot distinguish E2 and E3, the contradiction is obvious, as they have to commit $0$ in E2 and commit $1$ in E3.
\end{proof} 
    
    $\Delta+1.5\delta$ is a tight lower bound due to the matching (and surprising) upper bound.
    But we also provide some intuition below on why $\Delta+1.5\delta$ is the best lower bound we can prove using the above approach.
    Suppose we try to use the above construction to prove a lower bound of $\Delta+1.6\delta$, then in E1 and E2, $C$ have to start the protocol $0.6\delta$ time late; otherwise $g$ is able to distinguish E1 and E2 before its commit at time $\Delta+1.6\delta$. 
    Similarly in E4 and E3, $A$ have to start the protocol $0.6\delta$ time late.
    Then, in order to have E2 and E3 indistinguishable to $A$ and $C$, the message delays between $A,C$ must compensate for the $0.6\delta$ clock skew. 
    Since the message delay $l_3$ from $C$ to $A$ must be $\leq \Delta$ in E3, the message delay $l_2$ from $C$ to $A$ in E2 must be $l_2=l_3-2\times (0.6\delta)\leq \Delta-1.2\delta$. 
    The message delay from $A$ to $g$ in E2 is $\leq \delta$ in both E1 and E2. 
    Then, $g$ in E2 can learn that $C$ received $1$ from the broadcaster via the forwarded messages from $A$, before time $1.5\delta+l_2+\delta\leq \Delta+1.3\delta$. Thus, E1 and E2 are no longer indistinguishable to $g$.

\paragraph{BB upper bound $\Delta+1.5\delta$ under unsynchronized start and $f<n/2$.}
Now, we show the bound $\Delta+1.5\delta$ is tight under unsynchronized start, by presenting the protocol $(\Delta+1.5\delta)$-BB in Figure~\ref{fig:sync:D+1.5d}.

\textbf{Intuition.}
Before presenting the details of our protocol, we would like to provide the intuition of the state-of-the-art BB protocol with \latency $\Delta+2\delta$ from~\cite{abraham2020optimal}, and how our protocol improves the result to the optimal $\Delta+1.5\delta$, as illustrated in Figure~\ref{fig:intuition}.
The key insight of the $(\Delta+2\delta)$-BB is to use a $\Delta$ time waiting window for equivocation detection before voting, so that no two honest parties will vote for different values. More specifically, when receiving the proposed value from the broadcaster, the party forwards the proposal and waits for a time of $\Delta$. If no conflicting value is received during the $\Delta$ period, the party votes for the value. 
As shown in Figure~\ref{fig:intuition}, the forwarded value $v$ from the first honest party will reach all other honest parties within their $\Delta$ period and thus prevents them from voting for a different value $v'$. Since no two honest parties vote for different values, there can be at most one value with a certificate ($f+1$ votes), ensuring that all honest parties lock on the committed value.  
Further improving the \latency in this paradigm, however, is nontrivial. 
If we allow the parties to vote before the $\Delta$ period ends, there may be honest parties voting for different values before they detect equivocation. Then certificates for different values will be formed since $f$ Byzantine parties can double vote, and honest parties cannot tell which is the value that has been actually committed.

\begin{figure*}[t]
    \centering
    \resizebox{0.75\textwidth}{!}{
        \tikzset{every picture/.style={line width=0.75pt}} 

\begin{tikzpicture}[x=0.75pt,y=0.75pt,yscale=-1,xscale=1]

\draw    (300.38,60.19) -- (557.33,60.19) ;
\draw [shift={(560.33,60.19)}, rotate = 180] [fill={rgb, 255:red, 0; green, 0; blue, 0 }  ][line width=0.08]  [draw opacity=0] (8.93,-4.29) -- (0,0) -- (8.93,4.29) -- cycle    ;
\draw [line width=2.25]    (320.31,55.26) -- (320.31,60.4) ;
\draw [line width=2.25]    (360.64,55.6) -- (360.64,60.74) ;
\draw [line width=2.25]    (440.31,55.6) -- (440.31,60.74) ;
\draw [line width=2.25]    (520.31,55.6) -- (520.31,60.74) ;
\draw    (300.33,149.83) -- (557,149.83) ;
\draw [shift={(560,149.83)}, rotate = 180] [fill={rgb, 255:red, 0; green, 0; blue, 0 }  ][line width=0.08]  [draw opacity=0] (8.93,-4.29) -- (0,0) -- (8.93,4.29) -- cycle    ;
\draw [line width=2.25]    (320.38,145.13) -- (320.38,150.27) ;
\draw [line width=2.25]    (401.84,145.2) -- (401.84,150.34) ;
\draw [line width=2.25]    (481.51,145.2) -- (481.51,150.34) ;
\draw    (401.84,145.2) -- (518.68,61.9) ;
\draw [shift={(520.31,60.74)}, rotate = 504.51] [color={rgb, 255:red, 0; green, 0; blue, 0 }  ][line width=0.75]    (10.93,-3.29) .. controls (6.95,-1.4) and (3.31,-0.3) .. (0,0) .. controls (3.31,0.3) and (6.95,1.4) .. (10.93,3.29)   ;
\draw    (360.64,60.74) -- (466.79,148.23) ;
\draw [shift={(468.33,149.5)}, rotate = 219.5] [color={rgb, 255:red, 0; green, 0; blue, 0 }  ][line width=0.75]    (10.93,-3.29) .. controls (6.95,-1.4) and (3.31,-0.3) .. (0,0) .. controls (3.31,0.3) and (6.95,1.4) .. (10.93,3.29)   ;
\draw    (9.95,61.02) -- (237.33,61.02) ;
\draw [shift={(240.33,61.02)}, rotate = 180] [fill={rgb, 255:red, 0; green, 0; blue, 0 }  ][line width=0.08]  [draw opacity=0] (8.93,-4.29) -- (0,0) -- (8.93,4.29) -- cycle    ;
\draw [line width=2.25]    (29.88,56.1) -- (29.88,61.24) ;
\draw [line width=2.25]    (70.21,56.43) -- (70.21,61.57) ;
\draw [line width=2.25]    (189.88,56.43) -- (189.88,61.57) ;
\draw    (10.33,150.62) -- (237.33,150.62) ;
\draw [shift={(240.33,150.62)}, rotate = 180] [fill={rgb, 255:red, 0; green, 0; blue, 0 }  ][line width=0.08]  [draw opacity=0] (8.93,-4.29) -- (0,0) -- (8.93,4.29) -- cycle    ;
\draw [line width=2.25]    (30.28,145.3) -- (30.28,150.44) ;
\draw [line width=2.25]    (93.41,146.03) -- (93.41,151.17) ;
\draw [line width=2.25]    (213.08,146.03) -- (213.08,151.17) ;
\draw    (93.41,146.03) -- (214.02,62.97) ;
\draw [shift={(215.67,61.83)}, rotate = 505.45] [color={rgb, 255:red, 0; green, 0; blue, 0 }  ][line width=0.75]    (10.93,-3.29) .. controls (6.95,-1.4) and (3.31,-0.3) .. (0,0) .. controls (3.31,0.3) and (6.95,1.4) .. (10.93,3.29)   ;
\draw    (189.88,61.57) -- (227.72,89.32) ;
\draw [shift={(229.33,90.5)}, rotate = 216.25] [color={rgb, 255:red, 0; green, 0; blue, 0 }  ][line width=0.75]    (10.93,-3.29) .. controls (6.95,-1.4) and (3.31,-0.3) .. (0,0) .. controls (3.31,0.3) and (6.95,1.4) .. (10.93,3.29)   ;
\draw    (70.21,61.57) -- (172.15,149.2) ;
\draw [shift={(173.67,150.5)}, rotate = 220.68] [color={rgb, 255:red, 0; green, 0; blue, 0 }  ][line width=0.75]    (10.93,-3.29) .. controls (6.95,-1.4) and (3.31,-0.3) .. (0,0) .. controls (3.31,0.3) and (6.95,1.4) .. (10.93,3.29)   ;
\draw    (440.31,60.74) -- (467.65,78.41) ;
\draw [shift={(469.33,79.5)}, rotate = 212.88] [color={rgb, 255:red, 0; green, 0; blue, 0 }  ][line width=0.75]    (10.93,-3.29) .. controls (6.95,-1.4) and (3.31,-0.3) .. (0,0) .. controls (3.31,0.3) and (6.95,1.4) .. (10.93,3.29)   ;

\draw (357.17,41.5) node [anchor=north west][inner sep=0.75pt]  [font=\large] [align=left] {$\displaystyle t$};
\draw (409,42.17) node [anchor=north west][inner sep=0.75pt]  [font=\large] [align=left] {$\displaystyle t+\Delta -0.5d$};
\draw (495.33,42.33) node [anchor=north west][inner sep=0.75pt]  [font=\large] [align=left] {$\displaystyle t+\Delta +0.5d$};
\draw (315.87,41.7) node [anchor=north west][inner sep=0.75pt]  [font=\large] [align=left] {$\displaystyle 0$};
\draw (397.37,153.3) node [anchor=north west][inner sep=0.75pt]  [font=\large] [align=left] {$\displaystyle t'$};
\draw (453.1,152.77) node [anchor=north west][inner sep=0.75pt]  [font=\large] [align=left] {$\displaystyle t+\Delta -0.5d'$};
\draw (315.43,152.63) node [anchor=north west][inner sep=0.75pt]  [font=\large] [align=left] {$\displaystyle 0$};
\draw (501.3,71.5) node [anchor=north west][inner sep=0.75pt]  [font=\large] [align=left] {$\displaystyle \Delta $};
\draw (434.7,130.3) node [anchor=north west][inner sep=0.75pt]  [font=\large] [align=left] {$\displaystyle \Delta $};
\draw (179.07,43.33) node [anchor=north west][inner sep=0.75pt]  [font=\large] [align=left] {$\displaystyle t+\Delta $};
\draw (66.73,42.83) node [anchor=north west][inner sep=0.75pt]  [font=\large] [align=left] {$\displaystyle t$};
\draw (25.93,42.03) node [anchor=north west][inner sep=0.75pt]  [font=\large] [align=left] {$\displaystyle 0$};
\draw (198.17,153.93) node [anchor=north west][inner sep=0.75pt]  [font=\large] [align=left] {$\displaystyle t'+\Delta $};
\draw (87.93,153.63) node [anchor=north west][inner sep=0.75pt]  [font=\large] [align=left] {$\displaystyle t'$};
\draw (25.33,153.63) node [anchor=north west][inner sep=0.75pt]  [font=\large] [align=left] {$\displaystyle 0$};
\draw (173.93,72.47) node [anchor=north west][inner sep=0.75pt]  [font=\large] [align=left] {$\displaystyle \Delta $};
\draw (45,125) node [anchor=north west][inner sep=0.75pt]   [align=left] {{\large forward $\displaystyle v'$}};
\draw (229,72) node [anchor=north west][inner sep=0.75pt]   [align=left] {{\large vote $\displaystyle v$}};
\draw (32.88,72.74) node [anchor=north west][inner sep=0.75pt]   [align=left] {{\large forward $\displaystyle v$}};
\draw (158.27,122.13) node [anchor=north west][inner sep=0.75pt]  [font=\large] [align=left] {$\displaystyle \Delta $};
\draw (187.83,124.17) node [anchor=north west][inner sep=0.75pt]   [align=left] {{\large will not vote $\displaystyle v'$}};
\draw (322.81,71.9) node [anchor=north west][inner sep=0.75pt]   [align=left] {{\large forward $\displaystyle v$}};
\draw (355,124.67) node [anchor=north west][inner sep=0.75pt]   [align=left] {{\large forward $\displaystyle v'$}};
\draw (425.17,72.67) node [anchor=north west][inner sep=0.75pt]   [align=left] {{\large vote $\displaystyle v$}};
\draw (470,124.67) node [anchor=north west][inner sep=0.75pt]   [align=left] {{\large will not vote $\displaystyle v'$}};
\draw (81.5,183) node [anchor=north west][inner sep=0.75pt]   [align=left] {$\displaystyle ( \Delta +2\delta )$-BB~\cite{abraham2020optimal}};
\draw (382,183) node [anchor=north west][inner sep=0.75pt]   [align=left] {$\displaystyle ( \Delta +1.5\delta )$-BB};

\end{tikzpicture}
    }
    \vspace{-1em}
    \caption{Intuition of the $(\Delta+1.5\delta)$-BB Protocol.}
    \vspace{-1em}
    \label{fig:intuition}
\end{figure*}
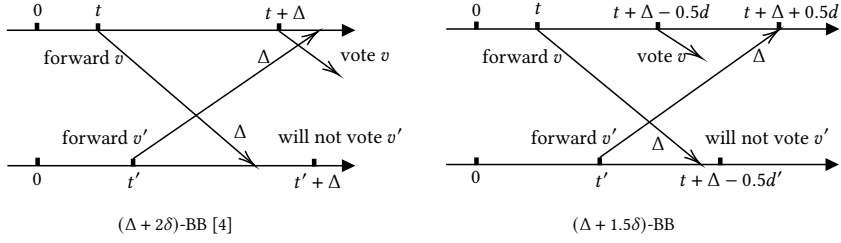

One novelty of our $(\Delta+1.5\delta)$-BB is to break such indistinguishability, by allowing parties to ``early vote'' with a parameter $d$ that ``guesses'' the value of $\delta$, and ranking the certificates by the value of $d$ (a smaller $d$ ranks higher).
So in our protocol, though honest parties may vote for different values, only the one with the highest rank will win, and we will guarantee that the certificate for any committed value always has the highest rank.
More specifically, for any value $d\in [0,\Delta]$, after $\Delta-0.5d$ time since receiving the proposed value $v$, parties send a vote containing $d$ and $v$, if no equivocation has been detected so far. Then, if $f+1$ votes with the same parameter $d$ and the same value $v$ are received, and no equivocation is detected for $\Delta+0.5d$ time since receiving $v$, a party can commit $v$. 
Our protocol guarantees that no honest party can vote for any other value $v'\neq v$ with a parameter $d'\leq d$ (Lemma~\ref{lem:bb3:1}). 
The intuition is that, as shown in Figure~\ref{fig:intuition}, if the second honest party receives the proposal $v'$ no later than some time threshold, its forwarded proposal will stop the first honest party from committing $v$. But if the second honest party receives $v'$ later than the time threshold, the forwarded proposal of $v$ from the first honest party will stop it from sending any votes with parameter $d'\leq d$ due to detecting equivocation.
Our construction guarantees a \latency of $\Delta+1.5\delta$. 
When the broadcaster is honest, all honest parties receive the value within time $\delta$, send vote with $d=\delta$ within time $\delta+\Delta-0.5\delta$, and receive $f+1$ votes from honest parties and commit within time $\delta+\Delta-0.5\delta+\delta=\Delta+1.5\delta$.
For votes with $d>\delta$, parties can only commit at time $\delta+\Delta+0.5d$ by our protocol, which leads to a latency $>\Delta+1.5\delta$.
For votes with $d<\delta$, $f+1$ such votes sent at time $\Delta-0.5d$ may not be received by all honest parties at time $\Delta+0.5d$, as the message delay $\delta>d$. 
%
It should be noted even if parties ``guess'' the value of $\delta$ wrong, the protocol always guarantees agreement, termination, validity, and the optimal \latency of $\Delta+1.5\delta$.

\begin{figure}[tb]
    \centering
    \begin{mybox}
    Initially, every party $i$ starts the protocol at most $\delta$ time apart with a local clock from $0$, sets $\texttt{direct-rcv}=false$, $\lock=\bot$, $\clockskew=\Delta$ and $\rank=\Delta+1$.

\begin{enumerate}
    \item\label{bb3:step:propose} \textbf{Propose.} The designated broadcaster $L$  with input $v$ sends $\langle \propose, v\rangle_L$ to all  parties.

    \item\label{bb3:step:forward} \textbf{Forward.} 
    When receiving the first valid proposal $\langle \propose, v\rangle_L$ from party $j$ at local time $t_{prop}$, forward the proposal to all parties.
    If $j=L$ and $t_{prop}\leq \Delta+\clockskew$, set $\texttt{direct-rcv}=true$.
    
    \item\label{bb3:step:vote} \textbf{Vote.}
    For every $d\in [0,\Delta]$~\footnote{The protocol is purely theoretical as the message complexity is unbounded. Its purpose is just to show the tightness of the  $\Delta+1.5\delta$ bound. 
    }
    , at local time $t_{prop}+\Delta-0.5d$, if the party detects no equivocation, it sends a \vote in the form of $\langle \vote, d, \langle \propose, v\rangle_L \rangle_i$ to all parties.

    \item\label{bb3:step:commit} \textbf{Commit and Lock.}
    When receiving $f+1$ signed \vote messages of the same $d$ and $v$ at local time $t_{votes}$, forwards these $f+1$ \vote messages and performs the following:
    \begin{enumerate}
        \item\label{bb3:step:commit:1} If $t_{votes}-t_{prop}\leq \Delta+1.5d$, detect no equivocation until local time $t_{prop}+\Delta+0.5d$, and $\texttt{direct-rcv}=true$, commit $v$.
        \item\label{bb3:step:commit:3} If $t_{votes}-t_{prop}\leq 4.5\Delta$ and $\rank > d$, update $\lock=v$ and $\rank=d$.
    \end{enumerate}

    \item\label{bb3:step:ba} \textbf{Byzantine agreement.}
    At local time $6.5\Delta+2\clockskew$, invoke an instance of Byzantine agreement with $\lock$ as the input. 
    If not committed, commit on the output of the Byzantine agreement. 
    Terminate.

\end{enumerate}
    \end{mybox}
    \caption{$(\Delta+1.5\delta)$-BB Protocol with $n/3 \leq f<n/2$ and unsynchronized start. }
    \vspace{-1em}
    \label{fig:sync:D+1.5d}
\end{figure}

\textbf{Protocol description.}
Each party starts the protocol at most $\clockskew$ time apart with a local clock starting at $0$, and it is guaranteed that $\clockskew\leq \delta$ by any clock synchronization protocol~\cite{dolev1995dynamic, abraham2019synchronous}.
Since the value of $\delta$ is unknown to the protocol designer or any party, all parties set the parameter $\clockskew=\Delta$ in the protocol.
Initially, each party sets $\lock=\bot$, $\rank=\Delta+1$, a flag $\texttt{direct-rcv}=false$, and starts the protocol at time at most $\clockskew=\Delta$ apart with its local clock starting from $0$.
The broadcaster first multicasts its proposed value, and all parties forward the first valid proposal received.
If the party receives the first valid proposal from the broadcaster at local time $t_{prop}$ and $t_{prop}\leq \Delta+\clockskew$, it sets the flag $\texttt{direct-rcv}=true$.
For every $d\in[0,\Delta]$, after $\Delta-0.5d$ time since the proposal is received, the party multicasts a \vote with parameter $d$ if no equivocation is detected so far.
When receiving $f+1$ \vote with the same value of $d$ 
and the same value $v$ at time $t_{votes}$, the party forwards these \vote messages and checks the following to commit or lock.
If the time between receiving the votes and the proposal is $\leq \Delta+1.5d$, and no equivocation is detected until local time $t_{prop}+\Delta+0.5d$, and the party receives the proposal from the broadcaster ($\texttt{direct-rcv}=true$), the party commits $v$.
If the time between receiving the votes and the proposal is $\leq 4.5\Delta$, and its $\rank>d$, the party updates its lock and rank.
Finally, at local time $6.5\Delta+2\clockskew$, the parties participate in an instance of BA with input $\lock$, and commit the output if have not committed.

\textbf{Correctness of the $(\Delta+1.5\delta)$-BB Protocol.}
\label{sec:ubproof:1+1.5}
In the proof, we use {\em local time} to refer the time indicated by the local clock at each party, and {\em global time} to refer the time indicated by some global clock.

\begin{lemma}\label{lem:bb3:1}
    If an honest party commits some value $v$ at Step~\ref{bb3:step:commit:1} by receiving $f+1$ \vote messages of the same value of $d\in[0,\Delta]$ and $v$, then 
    (1) no honest party sends \vote with $d'\leq d$ for any $v'\neq v$,
    (2) no honest party commits $v'\neq v$ at Step~\ref{bb3:step:commit}, and
    (3) all honest parties have $\lock=v$ when invoking the BA at Step~\ref{bb3:step:ba}.
    
\end{lemma}

\begin{proof}
    {\bf Part (1).}
    Suppose on the contrary that an honest party $h$ receives $f+1$ \vote messages of the same $d,v$ and then commits $v$ at Step~\ref{bb3:step:commit:1} at global time $t$, and some honest party $h'$ sends \vote  with $d'\leq d$ for some value $v'\neq v$.
    If $h'$ receives the proposal of $v'$ at global time $\leq t-\Delta$, then its forwarded proposal reaches $h$ at global time $\leq t$, and will stop $h$ from committing $v$ due to detecting equivocation. Hence $h'$ receives the proposal of $v'$ at global time $>t-\Delta$. 
    Since $h$ commits $v$ at local time $\geq t_{prop}+ \Delta+0.5d$ and at global time $t$, it receives the proposal of $v$ and forwards it to all parties at local time $t_{prop}$ and at global time $\leq t-\Delta-0.5d$. 
    Thus, the forwarded proposal of $v$ will reach $h'$ at global time $\leq t-\Delta-0.5d+\Delta=t-0.5d$.
    Since $h'$ receives the proposal of $v'$ at global time $> t-\Delta$ and local time $t_{prop}'$,  when $h'$ receives the forwarded proposal of $v$ from $h$ at global time $\leq t-0.5d$, its local time should be $<t_{prop}'+(t-0.5d)-(t-\Delta)=t_{prop}'+\Delta-0.5d\leq t_{prop}'+\Delta-0.5d'$ where $d'\leq d$. Therefore, $h'$ will not send any \vote with $d'\leq d$ for any value $v'\neq v$, since it detects equivocation within local time $t_{prop}'+\Delta-0.5d'$. This is a contradiction and thus no honest party sends \vote with $d'\leq d$ for any $v'\neq v$.
    
    {\bf Part (2).} 
    Suppose on the contrary that
    an honest party $h$ commits $v$ at Step~\ref{bb3:step:commit:1} at global time $t$, and some honest party $h'$ commits $v'\neq v$ at Step~\ref{bb3:step:commit:1}. Similar to Part (1), $h'$ must receive the proposal of $v'$ at global time $>t-\Delta$, otherwise $h$ will receive the proposal of $v'$ and not commit $v$. Also, the forwarded proposal of $v$ from $h$ will reach $h'$ at global time $\leq t-\Delta-0.5d+\Delta=t-0.5d$, thus the local time at $h'$ when receiving the proposal of $v$ is $<t_{prop}'+(t-0.5d)-(t-\Delta)=t_{prop}'+\Delta-0.5d\leq t_{prop}'+\Delta+0.5d'$ for any $d'\in[0,\Delta]$. Hence, $h'$ will not commit $v'$ at Step~\ref{bb3:step:commit} due to the detection of equivocation. This is a contradiction and thus no honest party commits $v'\neq v$ at Step~\ref{bb3:step:commit}.
    
    {\bf Part (3).}
    Any two honest parties receive the first valid proposal at most $\Delta$ time apart, since the first honest party that receives the proposal will forward it to all other parties, and the forwarded proposal will arrive within time $\Delta$.
    After the honest party $h$ commits some value $v$ at Step~\ref{bb3:step:commit:1} when $t_{votes}-t_{prop}\leq \Delta+1.5d$, all honest parties receives the $f+1$ \vote with $d$ messages forwarded by this honest party within $\Delta$ time.  
    Therefore, when receiving the $f+1$ \vote messages with $d$, any party has $t_{votes}-t_{prop}\leq \Delta+1.5d+2\Delta \leq 4.5\Delta$.
    
    Since an honest party $h$ commits $v$ at Step~\ref{bb3:step:commit:1}, it has $\texttt{direct-rcv}=true$ and receives the proposal from the broadcaster at local time $\leq \Delta+\clockskew$.
    Then, its forwarded proposal reaches all other parties at their local time $\leq 2\Delta+2\clockskew$, since the message delay is bounded by $\Delta$ and local clocks at any two parties have skew $\leq \clockskew$. 
    Since $t_{prop}\leq 2\Delta+2\clockskew$ at any party, any party has $t_{votes}\leq t_{prop}+4.5\Delta\leq 6.5\Delta+2\clockskew$, and thus will not invoke the BA at Step~\ref{bb3:step:ba} before setting the \lock.
    Moreover, by Part (1), no honest party sends any \vote with $d'\leq d$ for any $v'\neq v$, there exists no $f+1$ \vote with $d'\leq d$ for any $v'\neq v$. Hence all honest parties set $\lock=v$ at Step~\ref{bb3:step:commit:3} and will not change the $\lock$.
\end{proof}
\begin{theorem}\label{thm:ub:sync:unsyncstart}
    $(\Delta+1.5\delta)$-BB protocol solves Byzantine broadcast under $n/3<f<n/2$ faults in the synchronous authenticated setting, and has optimal \latency of  $\Delta+1.5\delta$.
\end{theorem}

\begin{proof}
    {\bf Agreement.}
    If all honest parties commit at Step~\ref{bb3:step:ba}, all honest parties commit on the same value due to the agreement property of the BA.
    Otherwise, there must be some honest party that commits at Step~\ref{bb3:step:commit}. By Lemma~\ref{lem:bb3:1}, no honest party commits $v'\neq v$ at Step~\ref{bb3:step:commit} and all honest parties set $\lock=v$ at Step~\ref{bb3:step:ba}. Since all honest parties input $v$ to the BA, by the validity of BA, the output of BA is $v$, so any honest party that has not committed will commit $v$.
    
    {\bf Termination.}
    According to the protocol, honest parties invoke a BA instance at local time $6.5\Delta+2\clockskew$, and terminate at Step~\ref{bb3:step:ba}. The parties commit a value before termination at Step~\ref{bb3:step:ba} or \ref{bb3:step:commit}.

    {\bf Validity.}
    If the broadcaster is honest, it sends the same proposal of value $v$ to all parties, and all honest parties receive the proposal at local time $\leq \Delta+\clockskew$ and set $\texttt{direct-rcv}=true$. 
    All $n-f\geq f+1$ honest parties will send \vote with $d$ for the proposal at Step~\ref{bb3:step:vote}, and there exists no \vote with $d$ for any $v'\neq v$. 
    Then at Step~\ref{bb3:step:commit}, all honest parties detect no equivocation and receive $f+1$ signed \vote with $d$ for $v$ from honest parties, thus commit $v$.
    
    {\bf Good-case latency.}
    In the good case, the broadcaster is honest and sends the same proposal of value $v$ at global time $0$. The proposal reaches all parties by global time $\delta$ and all honest parties set $\texttt{direct-rcv}=true$.
    Then, by global time $\delta+(\Delta-0.5\delta)$ all $n-f\geq f+1$ honest parties send \vote with $\delta$ 
    for $v$, and the above $n-f\geq f+1$ \vote with $\delta$ reach all honest parties at global time $\leq \delta+(\Delta-0.5\delta)+\delta=\Delta+1.5\delta$. 
    Since any honest party receives the proposal at global time $\geq 0$, we have $t_{votes}-t_{prop}\leq \Delta+1.5\delta$ at all honest parties. Moreover, honest parties detect no equivocation and have $\texttt{direct-rcv}=true$. 
    Thus, they all commit on $v$ at Step~\ref{bb3:step:commit:1} within time $\Delta+1.5\delta$ after the broadcaster sends the proposal. Therefore, the \latency of the protocol is $\Delta+1.5\delta$.
\end{proof}

\textbf{Tradeoff between communication complexity and good-case latency.}
The $(\Delta+1.5\delta)$-BB protocol has unbounded communication complexity to achieve the optimal good-case latency. In practice, we can bound the cost while achieving near-optimal good-case latency. More specifically, in Step~\ref{bb3:step:vote} of $(\Delta+1.5\delta)$-BB, we can discretely sample $m$ values for $d$ from the interval uniformly, to obtain near-optimal good-case latency $(1+\frac{1}{2m})\Delta+1.5\delta$ with communication cost $O(mn^2)$. Another practical protocol with good-case latency $\Delta+2\delta$ and cost $O(n^2)$ was proposed in~\cite{abraham2020optimal}.

\subsection{$n/2<f<n$, Lower Bound of $(\lfloor \frac{n}{n-f} \rfloor -1)\Delta$ and Upper Bound of $O(\frac{n}{n-f})\Delta$}
\label{sec:sync:5}

For the dishonest majority case, we prove a lower bound of $(\lfloor \frac{n}{n-f} \rfloor -1)\Delta$ on the \latency, and an upper bound of $O(\frac{n}{n-f})\Delta$ can be implied by the BB protocol in~\cite{constantroundBB}.
We defer the results to Appendix~\ref{sec:sync>1/2} for brevity.

\section{Related Work}

\paragraph{Improving worst-case latency for broadcast protocols.}
Byzantine fault tolerant broadcast, first proposed by Lamport et al. \cite{lamport1982byzantine}, have received a significant amount of attention for several decades.
For synchronous protocols, the deterministic Dolev-Strong protocol \cite{dolev1983authenticated} solves BB in worst-case $f+1$ rounds, matching a lower bound~\cite{fischer1982lower}.
For protocols with early stopping~\cite{dolev1990early}, a lower bound of $f'+2$ rounds exists for $f'<f$ actual faults. 
The classic asynchronous Byzantine reliable broadcast by Bracha~\cite{bracha1987asynchronous} has a worst-case latency of $3$ rounds.
A sequence of effort has been made on reducing the round complexity to expected constant through randomization \cite{ben1983another, rabin1983randomized,feldman1988optimal,katz2006expected,abraham2019synchronous}.

\paragraph{Improving \latency for BFT protocols.}
Decades of research on BFT state machine replication protocols focus on improving the performance of the protocol when an honest leader is in charge, which is what we formally defined as the \latency.
Under synchrony, Sync HotStuff \cite{synchotstuff} is a recent synchronous BFT SMR protocol that has a \latency of $2\Delta$. 
Later work~\cite{abraham2020optimal, abraham2020brief} improves the \latency to $\Delta+2\delta$ with a lower bound of $\Delta$ for the $n/3<f<n/2$ case. Our work closes the gaps, and gives a complete categorization of the \latency for synchronous and asynchronous broadcast protocols. 
Under partial synchrony, PBFT~\cite{castro1999practical} is a classic partially synchronous BFT SMR protocol with \latency of $3$ rounds and tolerates $f<n/3$ faults.
A sequence of works focus on improving the performance of PBFT, including FaB~\cite{martin2006fast}, Tendermint~\cite{buchman2016tendermint}, SBFT~\cite{gueta2019sbft}, HotStuff~\cite{yin2019hotstuff}.
FaB~\cite{martin2006fast} and 
a recent work~\cite{rambaud2020latency} 
prove \latency lower bound of $3$ rounds, for agreement problems with different validity guarantees. 
A concurrent work~\cite{kuznetsov2021revisiting} obtains results overlapping with our results for partially synchrony, with the problem formulation for agreement instead of broadcast. The lower bound result in~\cite{kuznetsov2021revisiting} is not limited to leader-based protocols and thus more general.

\paragraph{Optimistic BFT protocols.}
Another line of research aims at developing BFT protocols with small latencies when certain \emph{optimistic} conditions are met \cite{pass2018thunderella, synchotstuff,dutta2005best,song2008bosco,kotla2007zyzzyva}.
Common examples of such optimistic conditions include: more than $3n/4$ parties are honest in synchrony~\cite{pass2018thunderella, synchotstuff, shrestha2020optimality} or all $n$ parties vote under partial synchrony~\cite{kotla2007zyzzyva,gueta2019sbft}.
Note that these conditions are much more demanding than our definition of good-case, which only requires an honest leader.

\section{Conclusion and Open Problems}
We provide a complete categorization of the \latency of Byzantine fault-tolerant broadcast in the authenticated setting under synchrony, partial synchrony and asynchrony.
Our results reveal the structure in the latencies of Byzantine fault-tolerant broadcast with an honest broadcaster under various resilience assumptions, from which latency efficient state machine replication protocols can be derived.

The first open problem is the tight bound of the \latency for the $n/2\leq f <n$ case under synchrony.
Currently our result still leaves a gap of a constant factor (factor of $2$).
It would be interesting to complete the picture.

Another interesting open problem is to derive results for the unauthenticated case.
Some of the lower bound results in our paper still apply, but it is intriguing to find out if they are still tight. 
Under synchrony, unauthenticated BB is solvable if and only if $f<n/3$, and there exists a gap between the $2\delta$ lower bound and a $3\delta$ upper bound implied by Bracha's broadcast~\cite{bracha1987asynchronous}.
Under partial synchrony, we believe one can prove the tight resilience for unauthenticated \psyncbbshort with \latency of $2$ rounds is back to $n\geq 5f+1$ using our proof techniques. 
Under asynchrony, there exists a gap between the $3$-round upper bound by Bracha's broadcast~\cite{bracha1987asynchronous} protocol, and the $2$-round lower bound.

\paragraph{Acknowledgement.}
We would like to thank Jun Wan for helpful discussions.
We would like to thank Andrei Tonkikh for pointing out a subtle issue in the good-case latency definition under partial synchrony.

\bibliographystyle{ACM-Reference-Format}
\bibliography{references}

\clearpage

\appendix

\section{Asynchronous Model~\cite{canetti1993fast}}
\label{app:asyncmodel}
To measure the latency of an asynchronous protocol, we adopt the standard asynchronous round definitions from the literature~\cite{canetti1993fast}.

\begin{definition}[Asynchronous Atomic Step]
    The execution of the protocol proceeds in consecutive atomic steps, where for each atomic step the adversary can choose a single party $i$ to perform the following actions:
    \begin{itemize}
        \item Read a subset $\mathcal{M}$ of messages in its buffer chosen by the adversary. Messages in $\mathcal{M}$ are called delivered.
        \item Perform local computations.
        \item Send messages to other parties, and each message is buffered at the corresponding recipient until delivered.
    \end{itemize}
\end{definition}

\begin{definition}[Asynchronous Round]
    Each atomic step in an execution is assigned a round number as follows.
    \begin{itemize}
        \item Round $0$ only consists a single atomic step at each party, where each party receives a special \texttt{start} message to start the protocol. By definition, the message sent at this atomic step has round number $0$.
        
        \item For any $r\geq 1$, let $l_r$ be the last atomic step where a message of round $r-1$ is delivered. All the steps after step $l_{r-1}$ until (and including) step $l_{r}$ are in round $r$.
    \end{itemize}
\end{definition}

As an example, a Byzantine broadcast protocol that runs in $2$ asynchronous rounds can do the following.
In round $0$, the broadcaster multicasts its proposal to all other parties.
In round $1$, all parties receive the proposal from the broadcaster, and multicast vote messages to all other parties.
In round $2$, all parties receive enough vote messages and commit.
Note that only $2$ asynchronous rounds of message exchanges are needed for the protocol to commit, therefore the protocol has a commit latency of $2$ rounds.

\section{Missing Results and Proofs for Partial Synchrony}

\subsection{Correctness of the $(5f-1)$-\psyncbbshort Protocol}
\label{sec:psyncbb:proof}
For brevity, any value $v$ in the proof is assumed to be $v\neq \bot$, unless specified.

\begin{lemma}\label{lem:pbft:0}
    If $3f-1$ honest replicas vote for the same value $v$ in Step~\ref{pbft:step:vote} of view $w$, then there exists no valid certificate $\fc$ of view $w$ that locks any value  $v'\neq v$, and any honest replica that enters view $w+1$ has a valid certificate of view $w$ that locks $v$.
\end{lemma}

\begin{proof}
    Suppose on the contrary that there exists a valid certificate $\fc$ of view $w$ that locks value $v'\neq v$. By definition, $\fc$ contains $\geq 4f-1$ signed $\langle *, w \rangle$, and
    (a) 
    it contains $\geq 2f-1$ $\langle v', w \rangle$ signed by $L_{w}$ and the replica, and no other value signed by $L_{w}$, or 
    (b) 
    it contains $\geq 2f$ $\langle v', w \rangle$ signed by $L_{w}$ and replicas other than $L_w$.
    Condition (a) cannot be true: 
    Since $3f-1$ honest replicas only multicast $\langle v, w \rangle$ signed by the leader $L_{w}$ and the corresponding replica in the \timeout messages, $\fc$ cannot include these messages, which implies there need to be at least $3f-1+4f-1=7f-2>n$ replicas, contradiction.
    Condition (b) also cannot be true: 
    If $L_{w}$ is honest, then it will not propose and sign two different values, contradiction.
    If $L_{w}$ is Byzantine, then $\fc$ contains at most $f-1$ messages from Byzantine replicas and thus at least $3f$ messages from the honest replicas. Since $3f-1$ honest replicas only multicast $\langle v, w \rangle$ signed by $L_w$ and the corresponding replica, at least $2f$ messages must be for $v$, which implies that at most $4f-1-2f=2f-1$ messages can be for $v'$, contradiction.
    Therefore, there exist no valid certificate $\fc$ of view $w$ that locks any value $v'\neq v$.

    According to the protocol, any honest replica that enters view $w+1$ either receives $4f-1$ valid \timeout messages of view $w$ that contains only one value signed by $L_{w}$, or $4f-1$ valid \timeout messages from replicas other than $L_{w}$.
    For the first case, since $3f-1$ honest replicas only multicast signed $\langle v, w \rangle$, the $4f-1$ valid \timeout messages can contain at most $2f$ signed $\langle \bot, w \rangle$, and thus can form a valid certificate that locks $v$.
    For the second case: (i) If $L_{w}$ is honest, then only one value can be signed by $L_{w}$ and the claim follows from the first case; (ii) If $L_{w}$ is Byzantine, then $4f-1$ valid \timeout messages contains at most $f-1$ signatures from Byzantine replicas and thus at least $3f$ signatures from the honest replicas. Since $3f-1$ honest replicas only multicast signed $\langle v, w \rangle$, at least $2f$ signatures must be on $v$, and thus can form a valid certificate that locks $v$.
\end{proof}

\begin{lemma}\label{lem:pbft:1}
    If an honest replica commits $v$ in view $w$, for any $w'\geq w$, there exists no valid certificate of any view $w'$ that locks any other value $v'\neq v$, and no honest replica vote for any $v'\neq v$ in any view $w'+1$.
\end{lemma}

\begin{proof}
    
    Suppose $h$ is the honest replica that commits $v$ in view $w$.
    We prove the lemma by induction on the view numbers.
    
    {\em Base case of view $w$.}
    According to the protocol, $h$ receives $4f-1$ signed \vote messages for $v$ in view $w$, among which at least $3f-1$ \vote messages are sent by the honest replicas. 
    Since at least $3f-1$ honest replicas vote $v$, by Lemma~\ref{lem:pbft:0}, there exists no valid certificate $\fc$ of view $w$ that locks any value $v'\neq v$, and any honest replica that enters view $w+1$ has a valid certificate of view $w$ that locks $v$.
    According to Step~\ref{pbft:step:vote}, to have any honest replica vote for $v'\neq v$, $L_{w+1}$ needs to send the proposal $v'$ with either (a) a valid certificate of view $w$ that locks $v'$, which is impossible since there exist no valid certificate of view $w$ that locks any value $v'\neq v$; (b) or $4f-1$ valid \status message of view $w$ each with a valid certificate and among which the highest certificate locks $v'$, which is also impossible since any honest replica that enters view $w+1$ has a valid certificate of view $w$ that locks $v$, and there exists no valid certificate of view $w$ that locks $v'\neq v$.
    Hence, no honest replica vote for any $v'\neq v$ in view $w+1$, and the lemma is true for the base case.
    
    {\em Induction step.} Assume the induction hypothesis that the lemma is true for view $w,w+1,...,k-1$.
    We only need to prove that the lemma is also true for view $k$.
    Since all honest replicas can only vote for $v$ in view $k$ by induction hypothesis, according to Lemma~\ref{lem:pbft:0}, there exist no valid certificate of view $k$ that locks any value $v'\neq v$. The highest valid certificate at any honest replica locks on $v$, since any honest replica that enters view $w+1$ has a valid certificate of view $w$ that locks $v$ and there exists no valid certificate of view $\geq w$ that locks $v'\neq v$ that can update the certificate in Step~\ref{pbft:step:newview} at any honest replica.
    Then, according to Step~\ref{pbft:step:vote}, no honest replica will vote for any $v'\neq v$ in view $k+1$ since the leader cannot have a valid certificate of view $k$ that locks $v'$, or $4f-1$ valid \status message of view $k$ that contains a highest certificate that locks $v'$. 
    Hence the lemma is true by induction.
    
\end{proof}

\begin{theorem}[Agreement]\label{thm:pbft:safety}
    If an honest replica commits $v$, no honest commits any $v'\neq v$.
\end{theorem}

\begin{proof}
    Let $h$ be the first honest replica that commits, and $h$ commits $v$ in view $w$.
    Suppose on the contrary that another honest $h'$ commits a different value $v'\neq v$ in view $\geq w$. 
    If $h'$ commits in view $w$, according to Step~\ref{pbft:step:commit}, at least $3f-1$ honest replicas have voted for $v'$. Since at least $3f-1$ honest replicas need to vote $v$ for $h$ to commit and the total number of honest replicas is only $4f-1$, at least $(3f-1)+(3f-1)-(4f-1)=2f-1>0$ honest replicas need to vote for different values in the same view, contradiction.
    If $h'$ commits in view $w'\geq w+1$, according to Step~\ref{pbft:step:commit}, at least $3f-1$ honest replicas vote for $v'$.
    However, by Lemma~\ref{lem:pbft:1}, no honest replica vote for any $v'\neq v$ in any view $\geq w+1$, and therefore no honest can commit $v'\neq v$.

\end{proof}

\begin{theorem}[Termination]\label{thm:pbft:liveness}
    After GST, every replica eventually commits and terminates.
\end{theorem}

\begin{proof}
    Let view $w$ be the first view after GST that has an honest leader. 
    If no honest replica ever enters the view $w$, there exist no $4f-1$ valid \timeout messages from the honest replicas, and thus at least one honest replica must have committed before view $w$ and its forwarded $4f-1$ signed \vote will lead other honest replicas to commit as well.
    Otherwise, all honest replicas eventually receive $4f-1$ \timeout messages to enter view $w$.
    Then, any honest replica sends a \status message with $\fc$ to the leader $L_w$ in Step~\ref{pbft:step:newview}.
    The leader $L_w$ can receive $4f-1$ valid \status messages of view $w-1$ each with a valid $\fc$ that locks some value since there are $4f-1$ honest replicas. If the signatures from \timeout messages form a valid certificate of view $w-1$ that locks some value $v$, the leader proposes $v$ with the certificate, otherwise it proposes value $v$ that is locked by the highest valid certificate from $4f-1$ valid \status messages.
    By definition of the certificate check, the proposed value is externally valid, and according to Step~\ref{pbft:step:vote}, any honest replica will vote for the proposal of the leader.
    Then, in Step~\ref{pbft:step:commit} all honest replicas can receive $4f-1$ signed \vote messages of view $w$ for the same value $v$, and thus commit $v$.
    Also, $4\Delta$ time is sufficient for an honest leader to commit a value at all honest replicas before any honest replica timeout, since any two honest replicas enter the new view within $\Delta$ time of each other, and the sending of \status, \propose and \vote message each takes at most $\Delta$ time.
    Thus, no honest replica will timeout before voting in the view with an honest leader.
\end{proof}

\begin{theorem}[Validity]\label{thm:pbft:validity}
    If the designated broadcaster is honest and $GST=0$, then all honest parties commit the broadcaster's value.
\end{theorem}
\begin{proof}
    If the broadcaster is honest and the network is synchronous, all $4f-1$ honest replicas receive the same value from the broadcaster, and will vote for the same value. Then after $2$ rounds, all honest replicas receive $4f-1$ votes for the broadcaster's value, and commit the value.
    Otherwise, since any honest replica only vote for values that are externally valid, only externally valid values can be committed.
\end{proof}

\begin{theorem}[Good-case Latency]\label{thm:pbft:latency}
    When the network is synchronous and the leader is honest, the proposal of the leader will be committed within $2$ rounds.
\end{theorem}

\begin{proof}
    By the proof of Theorem~\ref{thm:pbft:validity}.
\end{proof}

\section{Missing Results and Proofs for Synchrony}

\subsection{$0<f<n/3$, Matching Lower and Upper Bounds of $2\delta$}
\label{sec:sync<1/3}

\paragraph{BRB lower bound $2\delta$ under synchronized start and $f>0$.}
The proof is very similar to that of Theorem~\ref{thm:lb:async:2}.

\begin{proof}[Proof of Theorem~\ref{thm:lb:sync:2d}.]
    Suppose there exists a protocol $\Pi$ with synchronized start that has \latency $<2\delta$, which means the honest parties can always commit before time $2\delta$ if the designated broadcaster is honest.
    Let party $s$ be the broadcaster, and divide the remaining $n-1$ parties into two groups $A,B$ each with $\geq 1$ party. 
    For brevity, we often use $A$ ($B$) to refer all the parties in $A$ ($B$).
    Consider the following three executions of $\Pi$. 
    All the executions constructed below have message delays equal to $\delta$.
\begin{enumerate}
    \item Execution  1. The broadcaster $s$ is honest, and sends $0$ to all parties.
    Since the broadcaster is honest, by validity, parties in $A,B$ will commit $0$ before time $2\delta$.
    
    \item Execution  2. The broadcaster $s$ is honest, and sends $1$ to all parties.
    Since the broadcaster is honest, by validity, parties in $A,B$ will commit $1$ before time $2\delta$.
    
    \item Execution  3. The broadcaster $s$ is Byzantine, it sends $0$ to parties in $A$ and $1$ to parties in $B$.
\end{enumerate}

    {\em Contradiction.} 
    Recall that the message delays are $\delta$ in all executions. The set of messages received by $A$ from $B$ before time $2\delta$ are sent by $B$ before time $\delta$, and thus is identical in Execution  $1$ and $3$ since the local history of $B$ is identical before receiving from the broadcaster in these two executions. 
    Therefore, the parties in $A$ cannot distinguish Execution  $1$ and $3$ before time $2\delta$,  and thus will commit $0$ in Execution  $3$.
    Similarly, the parties in $B$ cannot distinguish Execution  $2$ and $3$ before time $2\delta$, and will commit $1$ in Execution  $3$.
    However, this violates the agreement property of BRB, and therefore no such protocol $\Pi$ exists.
\end{proof}

The same proof also applies to an even weaker broadcast formulation named Byzantine consistent broadcast (BCB), where termination of all honest parties is required only when the broadcaster is honest.

\paragraph{BB upper bound $2\delta$ under unsynchronized start and $f<n/3$.}
We show a matching upper bound of $2\delta$ on the \latency for BB with the $2\delta$-BB Protocol presented in Figure~\ref{fig:sync:3f}. Here we give a brief description.
Initially each party has its $\lock$ set to some default value $\bot$, and starts the protocol at most $\clockskew=\Delta$ time part with a local clock starting at $0$.
First, the broadcaster multicasts its proposed value $v$, and each party will vote for the first valid proposal (in the correct format and signed by the broadcaster) and multicast a vote.
When a party receives $n-f$ votes on the same value $v$, it forwards these votes and sets its $\lock$ to be $v$. The party also commits $v$ if this happens before time $2\Delta+\delta$.
At local time $3\Delta+2\delta$, all parties participate in an instance of BA with $\lock$ as the input, and commit the output (if they haven't committed already). 

Note that synchrony and a known $\Delta$ are crucial to the correctness of the protocol; otherwise under partial synchrony, the resilience bound would be different (see Theorem~\ref{thm:psync:lb:3}).

\begin{figure}[h]
    \centering
    \begin{mybox}
    Initially, every party $i$ starts the protocol at most $\clockskew=\Delta$ time apart with a local clock and sets $\lock=\bot$.
\begin{enumerate}
    \item\label{bb1:step:propose} \textbf{Propose.} The designated broadcaster $L$  with input $v$ sends $\langle \texttt{propose}, v\rangle_L$ to all  parties.

    \item\label{bb1:step:vote} \textbf{Vote.} 
    When receiving the first valid proposal from the broadcaster, 
    send a vote to all parties in the form of $\langle \texttt{vote}, v \rangle_i$ where $v$ is the value of the proposal.
    
    \item\label{bb1:step:commit} \textbf{Commit.}
    When receiving $n-f$ signed votes for some value $v$ at local time $t$, forward these $n-f$ votes to all parties and set $\lock=v$. 
    If $t\leq 2\Delta+\clockskew$,  commit $v$.

    \item\label{bb1:step:ba} \textbf{Byzantine agreement.}
    At local time $3\Delta+2\clockskew$, invoke an instance of Byzantine agreement with $\lock$ as the input. 
    If not committed, commit on the output of the Byzantine agreement. 
    Terminate.

\end{enumerate}
    \end{mybox}
    \caption{$2\delta$-BB Protocol with $f<n/3$}
    \label{fig:sync:2d}
\end{figure}

\paragraph{Correctness of the $2\delta$-BB Protocol.}

\begin{theorem}\label{thm:ub:sync:2d}
    $2\delta$-BB protocol solves Byzantine broadcast under $f<n/3$ faults in the synchronous authenticated setting, and has optimal \latency of  $2\delta$.
\end{theorem}

\begin{proof}
    {\bf Agreement.}
    If all honest parties commit at Step~\ref{bb1:step:ba}, all honest parties commit on the same value due to the agreement property of the BA.
    Otherwise, there must be some honest party that commits at Step~\ref{bb1:step:commit}.
    First, no two honest parties can commit different values at Step~\ref{bb1:step:commit}.
    Otherwise, since they both receive $n-f$ signed votes, the two sets of votes intersect at at least $(n-f)+(n-f)-n=n-2f\geq f+1$ parties. This implies at least one honest party votes for different values, which cannot happen according to Step~\ref{bb1:step:vote}.
    Let $h$ denote the first honest party that commits at Step~\ref{bb1:step:commit}, and let $v$ denote the committed value.
    Since $h$ commits and forwards $n-f$ votes at local time $t\leq 2\Delta+\clockskew$, all honest parties set $\lock=b$ at their local time $\leq 3\Delta+2\clockskew$ before invoking the Byzantine agreement primitive at Step~\ref{bb1:step:ba}, since the clock skew is $\clockskew$ and message delay is bounded by $\Delta$.
    Therefore, at Step~\ref{bb1:step:ba}, all honest parties have the same input to the BA. Then by the validity condition of the BA primitive, the output of the agreement is also $b$. Any honest party that does not commit at Step~\ref{bb1:step:commit} will commit on value $b$ at Step~\ref{bb1:step:ba}. 
    
    {\bf Termination.}
    According to the protocol, honest parties terminate at Step~\ref{bb1:step:ba}, and they commit a value before termination.

    {\bf Validity.}
    If the broadcaster is honest, it sends the same proposal of value $v$ to all parties, and all honest parties will vote for $v$ before local time $\Delta+\clockskew$.
    Then at Step~\ref{bb1:step:commit}, all honest parties receive $n-f$ signed messages of $v$ before local time $2\Delta+\clockskew$, and commits $v$.
    
    {\bf Good-case latency.}
    In the good case, the broadcaster is honest, its value $v$ reaches all parties at time $\leq \delta$ and all honest parties will vote for $v$.
    Next, the above votes reach all honest parties after $\delta$ time,
    and all honest parties commit on the sender's proposal within time $\leq 2\delta$.
\end{proof}

\subsection{$f=n/3$, Matching Lower and Upper Bounds of $\Delta+\delta$}
\label{sec:sync=1/3}

\paragraph{BRB lower bound $\Delta+\delta$ under synchronized start and $f \geq n/3$.}
\label{sec:sync:lb:syncstart}

\begin{proof}[Proof of Theorem~\ref{thm:lb:sync:D+d}]
    Suppose there exists a protocol $\Pi$ with synchronized start that has \latency $<\Delta+\delta$, which means the honest parties can always commit before time $\Delta+\delta$ if the designated broadcaster is honest.
    Divide $n\leq 3f$ parties into three groups $A,B,C$ of size $\leq f$ each. 
    For brevity, we often use $A$ ($B,C$) to refer all the parties in $A$ ($B,C$).
    Let one party $s$ in group $C$ to be the broadcaster.
    Consider the following three executions of $\Pi$.
\begin{enumerate}
    \item Execution  1. The message delay bound is $\delta$, and the message delay is $\delta$ between all pairs of honest parties. The broadcaster $s$ is honest, and sends $0$ to all parties with message delay $\delta$. The parties in $B$ are Byzantine, but behave as honest except that they pretend the message delays between $B,A$ and $B,C$ are both $\Delta$. All other message delays are $\delta$. 
    Since the broadcaster is honest, by validity, parties in $A,C$ will commit $0$ before time $\Delta+\delta$.
    
    \item Execution  2. The message delay bound is $\delta$, and the message delay is $\delta$ between all pairs of honest parties. The broadcaster $s$ is honest, and sends $1$ to all parties with message delay $\delta$. The parties in $A$ are Byzantine, but behave as honest except that they pretend the message delays between $A,B$ and $A,C$ are both $\Delta$. All other message delays are $\delta$. 
    Since the broadcaster is honest, by validity, parties in $B,C$ will commit $1$ before time $\Delta+\delta$.
    
    \item Execution  3. The message delay bound is $\Delta$, and thus the adversary can control the message delay to be any value in $[0,\Delta]$ between any pair of honest parties. The broadcaster $s$ is Byzantine, it sends $0$ to parties in $A$ and $1$ to parties in $B$ both with message delay $\delta$. The parties in $C$ are also Byzantine, they behave the same as $C$ to $A$ from Execution  $1$, and the same as $C$ to $B$ from Execution $2$.
    The message delay between $A,B$ is $\Delta$, and all other message delays are $\delta$. 
\end{enumerate}

    {\em Contradiction.} 
    Notice that in Execution  $1$ and $3$, the message delay is $\delta$ between the broadcaster and $B$, and is $\Delta$ between $A$ and $B$. The set of messages received by $A$ from $B$ before time $\Delta+\delta$ are sent by $B$ before time $\delta$, and thus is identical in Execution  $1$ and $3$ since the state of $B$ is identical before receiving from the broadcaster in two executions. Moreover, $C$ behaves identically to $A$ in both executions.
    Therefore, the parties in $A$ cannot distinguish Execution  $1$ and $3$ before time $\Delta+\delta$,  and thus will commit $0$ in Execution  $3$.
    Similarly, the parties in $B$ cannot distinguish Execution  $2$ and $3$ before time $\Delta+\delta$, since they receive the identical set of messages from the parties in $A$, and thus will commit $1$ in Execution  $3$.
    However, this violates the agreement property of BRB, and therefore no such protocol $\Pi$ exists.
\end{proof}

\paragraph{BB upper bound $\Delta+\delta$ under unsynchronized start and $f= n/3$}

\label{sec:ubproof:f=n/3}

\begin{theorem}\label{thm:bb2}\label{thm:ub:sync:D+d}
    $(\Delta+\delta)$-$n/3$-BB protocol solves Byzantine broadcast under $f\leq n/3$ faults in the synchronous authenticated setting, and has optimal \latency of  $\Delta+\delta$.
\end{theorem}

\begin{proof}
    {\bf Agreement.}
    Honest parties may commit at Step~\ref{bb2:step:commit} or \ref{bb2:step:ba}. 
    
    First we prove that no two honest parties commit different values in Step~\ref{bb2:step:commit}. Suppose two honest $h,h'$ commit value $v,v'$ respectively in Step~\ref{bb2:step:commit}. 
    Without loss of generality, suppose that $h$ receives the proposal of $v$ from the broadcaster no later than $h'$. Then party $h'$ should receive the vote for $v$ from party $h$ during its $\Delta$ waiting period, and will not commit at Step~\ref{bb2:step:commit} due to the detection of conflicting votes. Thus any honest commit the same value in Step~\ref{bb2:step:commit}.
    
    Now consider any honest party $h$ that commits in Step~\ref{bb2:step:ba} before invoking the BA. Then $h$ receives two sets of $n-f$ conflicting votes, and the parties in the intersection of the two sets are Byzantine since they voted for different values. The set $\ff$  contains at least $(n-f)+(n-f)-n=n-2f\geq f$ parties, which means $h$ detects all $f$ Byzantine parties. Therefore, any commit message from any party not in $\ff$ must be from an honest party, and $h$ commits the same value as any honest party committed at Step~\ref{bb2:step:commit}.
    
    If all honest parties commit at Step~\ref{bb2:step:ba} after the BA, all honest parties commit the same value due to the agreement property of the BA.
    Otherwise, there must be some honest party that commits some value $v$ at Step~\ref{bb2:step:commit}.
    Since this honest party receives and forwards $n-f$ votes at local time $\leq 2\Delta+\clockskew$, all honest parties receive these votes at local time $\leq 3\Delta+2\clockskew$. If any honest party receives $n-f$ votes for only one value $v$, it sets $\lock=v$. Otherwise, the honest party detects all Byzantine parties, commits the same value $v$ and sets $\lock=v$ as argued previously.
    Therefore, at Step~\ref{bb2:step:ba}, all honest parties have the same input to the BA. Then by the validity condition of the BA primitive, the output of the agreement is also $v$. Any honest party that does not commit at Step~\ref{bb2:step:commit} will commit on value $v$ at Step~\ref{bb2:step:ba}. 
    
    {\bf Termination.}
    According to the protocol, honest parties terminate at Step~\ref{bb1:step:ba}, and they commit a value before termination.

    {\bf Validity.}
    If the broadcaster is honest, it sends the same proposal of value $v$ to all parties, and all $n-f$ honest parties will vote for the proposal within $\Delta$ time. Moreover, there exists no vote for any other proposal. 
    Then at Step~\ref{bb1:step:commit}, all honest parties detect no conflicting vote and receive $n-f$ signed votes for $v$ from honest parties before time $2\Delta+\clockskew$, and commits $v$.
    
    {\bf Good-case latency.}
    In the good case, the broadcaster is honest, its proposal of value $v$ reaches all parties within time $\leq \delta$ and all honest parties will vote for $v$ and start the $\Delta$ waiting period. 
    The above $n-f$ votes reach all honest parties after $\delta$ time, and meanwhile the honest parties detect no conflicting vote during the $\Delta$ waiting period, thus they will commit on the sender's proposal at time $\leq \Delta+\delta$ in Step~\ref{bb2:step:commit}.
\end{proof}

\subsection{$n/3<f<n/2$ and Synchronized Start, Matching Lower and Upper Bounds of $\Delta+\delta$}

\paragraph{Correctness of the $(\Delta+\delta)$-BB Protocol.}
\label{sec:ubproof:1+1}
Since the protocol assumes synchronized start, local clocks at all honest parties have the same time.

\begin{lemma}\label{lem:bb4:1}
    If an honest party commits some value $v$ at Step~\ref{bb4:step:commit}, then (1) no honest party commits $v'\neq v$ at Step~\ref{bb4:step:commit}, and (2) all honest parties have $\lock=v$ at Step~\ref{bb4:step:ba}.
\end{lemma}

\begin{proof}
    {\bf Part (1).}
    Suppose an honest party $h$ commits some value $v$ at Step~\ref{bb4:step:commit} at time $\Delta+t$. According to the protocol, $h$ detects no equivocation within time $\Delta+t$ and receives $f+1$ signed \vote messages for the same value $v$ each with $d\leq t$. Then, at least one \vote is from an honest party $j$, who receives the proposal of $v$ and multicasts the \vote at time $d_j$ where $d_j\leq t$.
    Now suppose some honest party $h'$ commits a different value $v'\neq v$ at Step~\ref{bb4:step:commit} at time $\Delta+t'$. 
    Similarly, at least one honest party $k$ receives the proposal of $v'$ and multicasts the \vote at time $d_k$ where $d_k\leq t'$.
    Without loss of generality, assume $d_j\leq d_k$. Then the \vote from party $j$ will reach party $k$ within time $d_j+\Delta\leq d_k+\Delta \leq t'+\Delta$, and prevents $h'$ from committing due to equivocation detection.
    This is a contradiction, and thus no honest party commits $v'\neq v$ at Step~\ref{bb4:step:commit}.
    
    {\bf Part (2).}
    Suppose an honest party $h$ commits some value $v$ at Step~\ref{bb4:step:commit} at time $\Delta+t$. 
    Since $h$ forwards the $f+1$ \vote messages at time $\Delta+t$, all honest parties will receive these \vote messages within time $2\Delta+t$.
    According to the protocol, all honest parties will update their $\lock=v$, unless they are locked on some value $v'\neq v$ with $\rank\leq t$.
    Hence, it is sufficient to show that there exists no $f+1$ \vote messages for any $v'\neq v$ each has $d\leq t$.
    First we show that there is no \vote with $d\leq t$ for any $v'\neq v$ sent by an honest party. Otherwise, this \vote will reach party $h$ within time $d+\Delta\leq t+\Delta$ and prevent $h$ from committing $h$ due to equivocation detection. 
    Since there are at most $f$ Byzantine parties, there exists no $f+1$ \vote messages for any $v'\neq v$ each has $d\leq t$.
\end{proof}

\begin{theorem}\label{thm:ub:sync:syncstart}
    $(\Delta+\delta)$-BB protocol solves Byzantine broadcast tolerating $n/3<f<n/2$ faults under synchronized start in the synchronous authenticated setting, and has optimal \latency of  $\Delta+\delta$.
\end{theorem}

\begin{proof}
    {\bf Agreement.}
    If all honest parties commit at Step~\ref{bb4:step:ba}, all honest parties commit on the same value due to the agreement property of the BA.
    Otherwise, there must be some honest party that commits at Step~\ref{bb4:step:commit}. By Lemma~\ref{lem:bb4:1}, no honest party commits $v'\neq v$ at Step~\ref{bb4:step:commit} and all honest party have $\lock=v$ at Step~\ref{bb4:step:ba}. Since all honest parties have the same input $v$ for the BA, according to the validity condition of the BA, the output of the agreement is $v$. Then any honest party that has not committed will commit $v$.
    
    {\bf Termination.}
    According to the protocol, honest parties invoke a BA instance at time $4\Delta$, and terminate after the BA at Step~\ref{bb4:step:ba}. The parties commit a value before termination at Step~\ref{bb4:step:ba} or \ref{bb4:step:commit}.

    {\bf Validity and \latency.}
    If the broadcaster is honest, it sends the same proposal of value $v$ to all parties, and all honest parties receive the proposal within time $\delta$.
    All $n-f\geq f+1$ honest parties will send \vote with $d\leq \delta$ for the proposal at Step~\ref{bb4:step:vote}, and there exists no valid \vote for any $v'\neq v$. 
    Then at Step~\ref{bb4:step:commit}, within time $\Delta+\delta$, all honest parties detect no equivocation and receive $f+1$ signed \vote messages for $v$ with $d\leq \delta$ from the honest parties, and commit $v$.
    Thus the \latency of the protocol is $\Delta+\delta$.
\end{proof}

\subsection{$n/3<f<n/2$ and Unsynchronized Start, Matching Lower and Upper Bounds of $\Delta+1.5\delta$}

\paragraph{BRB lower bound $\Delta+1.5\delta$ under unsynchronized start and $ f>n/3$}
\label{sec:lbproof:1+1.5}

\begin{proof}[Proof of Theorem~\ref{thm:lb:sync:D+1.5d}]
    
The proof is illustrated in Figure~\ref{fig:BB-lb}.
Assume there exists a BRB protocol $\Pi$ that has \latency $<\Delta+1.5\delta$, under clock skew $\clockskew$ and $f>n/3$ Byzantine faults.
As mentioned in Section~\ref{sec:prelim}, we assume the clock skew $\clockskew= 0.5\delta$ due to the lower bound for clock skew~\cite{attiya2004distributed}.
We divide $n$ parties into $5$ disjoint groups $A,B,C,g,h$, where $g,h$ each contain a single party, and $A,B,C$ each evenly contains the remaining $n-2<3f-2$ parties so that each of $A,B,C$ contains $\leq f-1$ parties.

We construct four executions below. 
The honest parties have the same initial state at any executions constructed.
When any party starts the protocol, it also starts the local clock at the same time.
For brevity, we often use the group to refer all the parties in that group.

\begin{figure}
    \centering
    \includegraphics[width=0.75\textwidth]{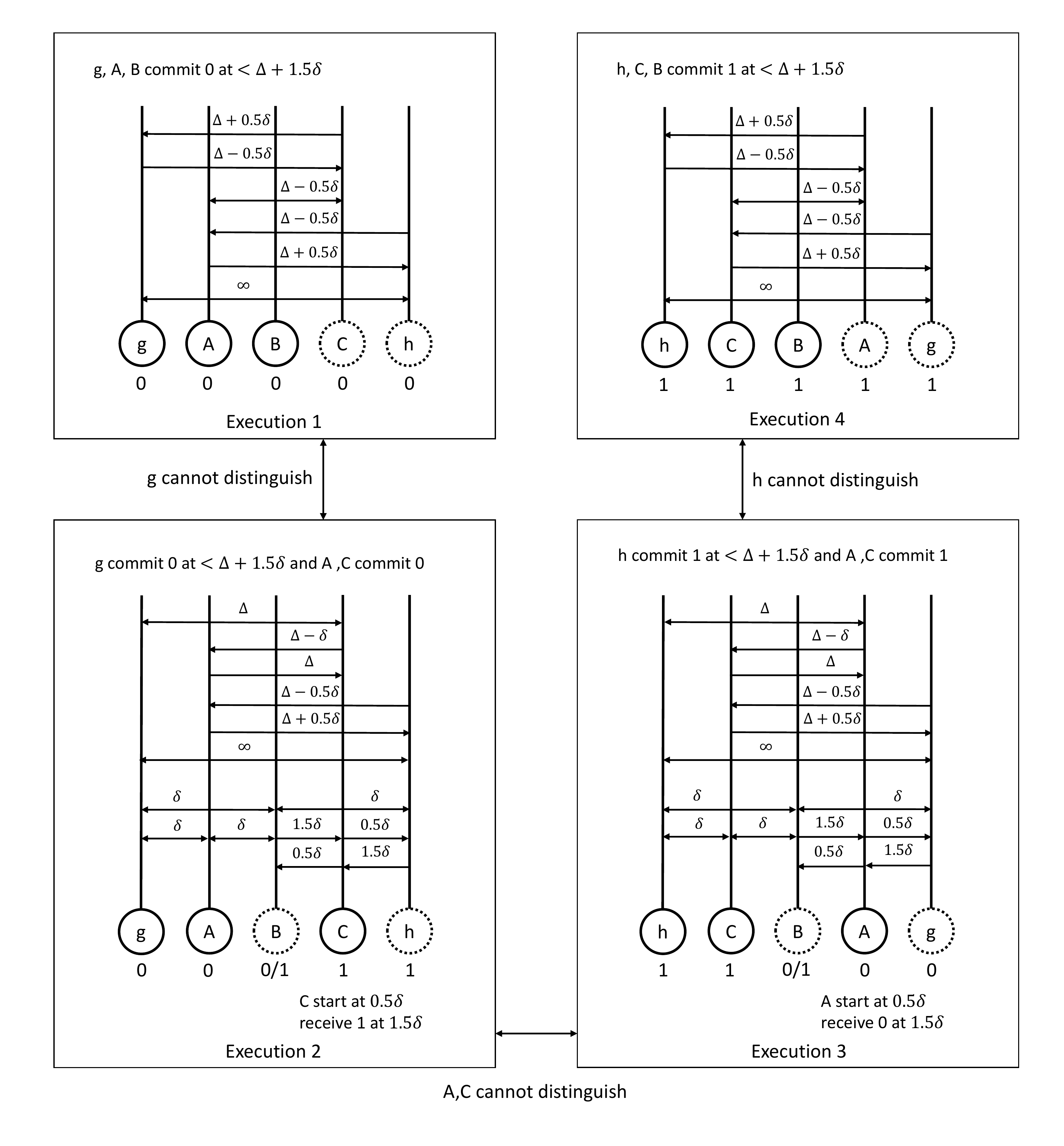}
    \caption{BRB \latency lower bound: $\Delta+1.5\delta$ with unsynchronized start. 
    Dotted circles denote Byzantine parties.
    }
    \label{fig:BB-lb}
\end{figure}

\begin{enumerate}[leftmargin=*]
    \item Execution  1. The message delay bound is $\delta$, and the message delay is $\delta$ between all pairs of honest parties.  The broadcaster is honest and sends value $0$ to all parties. All parties start the protocol at time $0$, and receive the broadcaster's proposal at time $\delta$.
    \begin{itemize}
        \item $g,A,B$ are honest.
        \item Parties in $C$ are Byzantine. They behave honestly except that they control the message delays to be $\Delta+0.5\delta$ from $C$ to $g$, $\Delta-0.5\delta$ from $C$ to $A$, and $\Delta-0.5\delta$ from $g,A$ to $C$.
        \item $h$ is Byzantine, it behaves honestly except that it controls the message delay to be $\Delta-0.5\delta$ from $h$ to $A$, $\Delta+0.5\delta$ from $A$ to $h$, and $\infty$ between $h$ and $g$.
    \end{itemize}
    Since the broadcaster is honest, the honest parties in 
    $g,A,B$ will commit $0$ at time $<\Delta+1.5\delta$ by assumption.
    
    \item Execution  4.  This Execution is a symmetry of Execution 1. The message delay is $\delta$ between all pairs of honest parties. The broadcaster is honest and sends value $1$ to all parties. All parties start the protocol at time $0$, and receive the broadcaster's proposal at time $\delta$.
    \begin{itemize}
        \item $h,C,B$ are honest.
        \item Parties in $A$ are Byzantine. They behave honestly except that they control the message delays to be $\Delta+0.5\delta$ from $A$ to $h$, $\Delta-0.5\delta$ from $A$ to $C$, and $\Delta-0.5\delta$ from $h,C$ to $A$.
        \item $g$ is Byzantine, it behaves honestly except that it controls the message delay to be $\Delta-0.5\delta$ from $g$ to $C$, $\Delta+0.5\delta$ from $C$ to $g$, and $\infty$ between $g$ and $h$.
    \end{itemize}
    Since the broadcaster is honest, the honest parties in 
    $B,C,h$ will commit $1$ within time $\Delta+1.5\delta$ by assumption.
    
    \item Execution  2. The message delay bound is $\Delta$, and thus the adversary can control the message delay to be any value in $[0,\Delta]$ between any pair of honest parties.
    The broadcaster is Byzantine, it sends $0$ to $g,A,B$, and $1$ to $B,C,h$. 
    Let $\delta$ denote the message delay bound in Execution $1$ and $4$.
    All parties start the protocol at time $0$ except that the parties in $C$ start at time $0.5\delta$.
    All parties receive the broadcaster's proposal at time $\delta$, except that $C$ receive at time $1.5\delta$.
    \begin{itemize}
        \item $g,A,C$ are honest.
        The message delay is $\delta$ between $g,A$, and $\Delta$ between $g,C$.
        The message delay is $\Delta-\delta$ from $C$ to $A$, and $\Delta$ from $A$ to $C$. 
        \item The parties in $B$ are Byzantine, they behave to $g,A$ the same as that in Execution 1, and to $C,h$ the same as that in Execution 4, but they control the message delay to be $1.5\delta$ from $B$ to $C$ and $0.5\delta$ from $C$ to $B$.
        \item $h$ is Byzantine, it behaves honestly except it controls the message delays as follows:
        $\infty$ between $g,h$, $\delta$ between $B,h$, $0.5\delta$ from $C$ to $h$, $1.5\delta$ from $h$ to $C$, $\Delta+0.5\delta$ from $A$ to $h$, and $\Delta-0.5\delta$ from $h$ to $A$.
    \end{itemize}
    
    \item Execution  3. This Execution is the symmetry of Execution 2.
    The message delay bound is $\Delta$ and let $\delta$ denote the message delay bound in Execution $1$ and $4$.
    The broadcaster is Byzantine, it sends $0$ to $g,A,B$, and $1$ to $B,C,h$.
    All parties start the protocol at time $0$ except that the parties in $A$ start at time $0.5\delta$.
    All parties receive the broadcaster's proposal at time $\delta$, except that $A$ receive at time $1.5\delta$.
    \begin{itemize}
        \item $h,A,C$ are honest.
        The message delay is $\delta$ between $h,C$, and $\Delta$ between $h,A$.
        The message delay is $\Delta-\delta$ from $A$ to $C$, and $\Delta$ from $C$ to $A$.
        \item The parties in $B$ are Byzantine, they behave to $g,A$ the same as that in Execution 1 but control the message delay to be $1.5\delta$ from $B$ to $A$ and $0.5\delta$ from $A$ to $B$, and to $C,h$ the same as that in Execution 4.
        \item $g$ is Byzantine, it behaves honestly except it controls the messages delays as follows:
        $\infty$ between $g,h$, $\delta$ between $B,g$, $0.5\delta$ from $A$ to $g$, $1.5\delta$ from $g$ to $A$, $\Delta+0.5\delta$ from $C$ to $g$, and $\Delta-0.5\delta$ from $g$ to $C$.
    \end{itemize}

\end{enumerate}

{\bf Claim 1: The party in $g$ cannot distinguish Execution 1 and 2 before time $\Delta+1.5\delta$, and thus it commits $0$ in both executions. Similarly, $h$ cannot distinguish Execution 3 and 4 before time $\Delta+1.5\delta$, and it commits $1$ in both executions.}
    
    We need to show that before $\Delta+1.5\delta$, the party in $g$ receives the same set of messages at the same corresponding time points by its local clock in both Execution 1 and 2. According to the construction of the executions, we have the following observation.
    \begin{itemize}[leftmargin=*]
        \item By construction, $B,h$ behave to $g$ the same in both Execution 1 and 2.
        \item For $A$, we show that the local history of $A$ before time $\Delta+0.5\delta$ is identical in both Execution 1 and 2. Then, any message from $A$ to $g$ before time $\Delta+1.5\delta$ is identical in both executions since the message delay between $g,A$ is $\delta$ in both executions.
        \begin{itemize}
            \item $A$ start the protocol at time $0$ with the same initial state, and receive proposal of value $0$ at time $\delta$ from the sender in both executions.
            \item By construction, $B$ behave the same to $A$ in both Execution 1 and 2.
            \item The messages from $C$ to $A$ before time $\Delta+0.5\delta$ are identical in Execution 1 and 2. 
            In Execution 1, $C$ start at time $0$, receive from the broadcaster at time $\delta$, and behave honestly but simulate a message delay of $\Delta-0.5\delta$ to $A$.  
            In Execution 2, $C$ start at time $0.5\delta$, receive from the broadcaster at time $1.5\delta$, and behave honestly. 
            Since the message delay from $C$ to $A$ is $\Delta-0.5\delta$ in Execution 1 and $\Delta-\delta$ in Execution 2 ($0.5\delta$ faster than Execution 1), and the local history at $C$ in Execution 2 from time $0.5\delta$ to $1.5\delta$ is identical to that in Execution 1 from time $0$ to $\delta$, the set of messages $g$ received from $C$ before time $\Delta+0.5\delta$ is the same in both executions.
            \item The message delay from $h$ to $A$ is $\Delta-0.5\delta$ in both executions, and $h$ has the same local history before time $\delta$. Hence any message from $h$ to $A$ received before time $\Delta+0.5\delta$ is the same in both executions.
            \item $g$ is honest in both executions, and will behave identically to $A$ before time $\Delta+1.5\delta$ unless different messages are received.
        \end{itemize}
        
        \item For $C$, we show that the messages $g$ received from $C$ before $\Delta+1.5\delta$ are identical in both Execution 1 and 2. 
        The local history of $C$ from time $0$ to time $<\delta$ in Execution 1 is identical to that from time $0.5\delta$ to time $<1.5\delta$ in Execution 2, since $C$ receive from the broadcaster at time $\delta$ in Execution 1, and at time $1.5\delta$ in Execution 2.
        Moreover, the message delay from $C$ to $g$ is $\Delta+0.5\delta$ in Execution 1, and $\Delta$ in Execution 2, which implies the claim.

    \end{itemize}
    By the argument above, the party in $g$ cannot distinguish Execution 1 and 2 before time $\Delta+1.5\delta$, and thus it commits $0$ in both executions.
    Similarly, $h$ cannot distinguish Execution 3 and 4 before time $\Delta+1.5\delta$, and thus it commits $1$ in both executions.

{\bf Claim 2: The parties in $A,C$ cannot distinguish Execution 2 and 3.}

We will prove that the local histories at $A,C$ are identical in both Execution 2 and 3.
By construction, $C$ in Execution 2 start the protocol and receive from the broadcaster $0.5\delta$ time later than $C$ in Execution 3, and  $A$ in Execution 3 start the protocol and receive from the broadcaster $0.5\delta$ time later than $A$ in Execution 2.
As for the message delays between $A,C$, the delay from $A$ to $C$ is $\Delta$ in Execution 2 and $\Delta-\delta$ in Execution 3, and the delay from $C$ to $A$ is $\Delta$ in Execution 3 and $\Delta-\delta$ in Execution 2. The differences in the message delays between $A,C$ compensate the delays of when $A,C$ start their protocol, and therefore, if any other receiving events at $A,C$ are identical from their local view, then $A,C$ will have the same local histories.
For other message delays, for $C$ in Execution 2 and $A$ in Execution 3 that start $0.5\delta$ time later, all incoming message delays are $0.5\delta$ larger and outgoing message delays are $0.5\delta$ smaller, which also compensate the delays of when $A,C$ start their protocol.
Rest of the message delays are identical in both executions.
The receiving events at $g,A,C,h$ from $B$ are identical in both executions, since $B$ behave to $g,A$ the same as that in Execution 1, and to $h,C$ the same as that in Execution 2.
The receiving events at $A,C$ from $g,h$ are also identical in both executions, since $g,h$ behave honestly except that they control delays to compensate the delay of when $A,C$ start their protocol.
Hence the local histories at $A,C$ are identical, and thus $A,C$ cannot distinguish Execution 2,3.

{\bf Contradiction.}
By Claim $1$, $g$ commits $0$ in Execution 2 and $h$ commits $1$ in Execution 3.
To satisfy safety, $A$ must commit $0$ in Execution 2 and $C$ must commit $1$ in Execution 3.
However, by Claim $2$, $A,C$ cannot distinguish Execution 2 and 3, and they may commit different values in the same execution, violating safety.
Therefore, such a BRB protocol $\Pi$ that has \latency $<\Delta+1.5\delta$ cannot exist.
\end{proof}


\subsection{$n/2<f<n$, Lower Bound of $(\lfloor \frac{n}{n-f} \rfloor -1)\Delta$ and Upper Bound of $O(\frac{n}{n-f})\Delta$}
\label{sec:sync>1/2}

For the dishonest majority case, we prove a lower bound of $(\lfloor \frac{n}{n-f} \rfloor -1)\Delta$ on the \latency, and an upper bound of $O(\frac{n}{n-f})\Delta$ can be implied by the BB protocol in~\cite{constantroundBB}.

\paragraph{BRB lower bound $(\lfloor \frac{n}{n-f} \rfloor -1)\Delta$ under synchronized start and $f\geq n/2$.}

We first show that no BRB protocol can have \latency less than $(\lfloor \frac{n}{n-f} \rfloor -1)\Delta$ under $f\geq n/2$.
The proof is inspired by the round complexity lower bound proof of Byzantine broadcast in \cite{garay2007round}, where the authors show that even randomized BB protocols require at least $2n/(n-f)-1$ rounds to terminate.

\begin{theorem}\label{thm:lb:sync:nD}
     Any Byzantine reliable broadcast protocol that is resilient to $f\geq n/2$ faults must have a \latency at least $(\lfloor \frac{n}{n-f} \rfloor -1)\Delta$, even with synchronized start.
\end{theorem}

\begin{proof}
The proof is illustrated in Figure~\ref{fig:BB-dishonest-majority-lb}.
Let $h=n-f$ denote the number of honest parties.
Let $d=2\lfloor \frac{n}{h} \rfloor-1$, which is odd.
We divide the parties into $d+1=2\lfloor \frac{n}{h} \rfloor$ disjoint groups $G_0,...,G_d$, where $|G_i|=1$ for $i=0,2,...,d-1$, $|G_i|=h-1$ for $i=1,3,...,d-2$ and $|G_d|\geq h-1$.
Suppose there exists a BRB protocol $\Pi$ that can tolerate $f$ Byzantine faults and commit before time $(\lfloor \frac{n}{n-f} \rfloor -1)\Delta=(d-1)\Delta/2$ when the broadcaster is honest.
For brevity, we often use group $G_i$ to refer all the parties in $G_i$.

Considering the following executions with the party in $G_0$ being the broadcaster. 
In all the executions below, 
any Byzantine party in group $G_i$ behaves as honest except that it only communicates with parties in groups $G_i$, $G_{i-1}$ and $G_{i+1}$. 
For the broadcaster in $G_0$, when it is Byzantine, it only communicates with $G_1$ and $G_d$ after sending the proposal.
\begin{itemize}
    \item Execution $0$. Only the broadcaster $G_0$ and $G_1$ are honest, and $G_0$ sends $0$ to all parties. All Byzantine parties in $G_2,...,G_d$ pretends their message delay is $\Delta$. Since the broadcaster is honest, $G_0,G_1$ commit $0$ before time $(d-1)\Delta/2$.
    
    \item Execution $d$. Only the broadcaster $G_0$ and $G_d$ are honest, and $G_0$ sends $1$ to all parties. All Byzantine parties in $G_1,...,G_{d-1}$ pretends their message delay is $\Delta$. Since the broadcaster is honest, $G_0,G_d$ commit $1$ before time $(d-1)\Delta/2$.
    
    \item Execution $i$, where $i=1,2,...,d-1$. Only $G_i, G_{i+1}$ are honest. The broadcaster in $G_0$ is Byzantine, it sends $0$ to $G_j$ for $1\leq j\leq (d+1)/2$, and sends $1$ to $G_j$ for $(d+1)/2\leq j\leq d+1$. $G_0$ behaves to $G_1$ the same as $G_0$ to $G_1$ in Execution $0$, and behaves to $G_d$ the same as $G_0$ to $G_d$ in Execution $d$. All other message delays are $\Delta$.
    \item In the executions above, we argue the following indistinguishability:
    \begin{itemize}
        \item $G_1$ cannot distinguish Execution $0$ and $1$ before time $(d-1)\Delta/2$.
        Any message sent by $G_{(d+1)/2}$ takes time $(d-1)\Delta/2$ to reach $G_1$. Before that the set of messages received by $G_1$ is identical in both executions, since $G_1,G_2,...,G_{(d-1)/2}$ all receive $0$ from the broadcaster, and any different message takes $(d-1)\Delta/2$ to reach $G_1$. 
        Thus, $G_1$ also commits $0$ in Execution $1$ before time $(d-1)\Delta/2$.
        Similarly, $G_d$ cannot distinguish Execution $d$ and $d-1$, and commits $1$ in Execution $d-1$ before time $(d-1)\Delta/2$.
        \item $G_i$ cannot distinguish Execution $i-1$ and Execution $i$ for $2\leq i \leq d-1$, since two executions look identical to $G_i$. The only difference between Execution $i-1$ and Execution $i$ is that the set of honest parties changes from $G_{i-1},G_i$ to $G_i,G_{i+1}$, but since the Byzantine parties behaves as honest except that they only communicate with neighboring parties, $G_i$ cannot distinguish the two executions.
    \end{itemize}
    {\em Contradiction.} 
    Since $G_i$ cannot distinguish Execution $i-1$ and Execution $i$ for $2\leq i \leq d-1$, by the termination property of BRB, we can infer that $G_{(d+1)/2}$ commit $0$ in Execution $(d-1)/2$ and commit $1$ in Execution $(d+1)/2$. However, $G_{(d+1)/2}$ cannot distinguish these two executions, and thus will violate agreement property of Byzantine broadcast. Hence, any BRB protocol with good case latency $<(d-1)\Delta/2=(\lfloor \frac{n}{n-f} \rfloor -1)\Delta$ cannot exist. \qedhere
\end{itemize}
\end{proof}

\begin{figure}[t]
    \centering
    \includegraphics[width=0.6\textwidth]{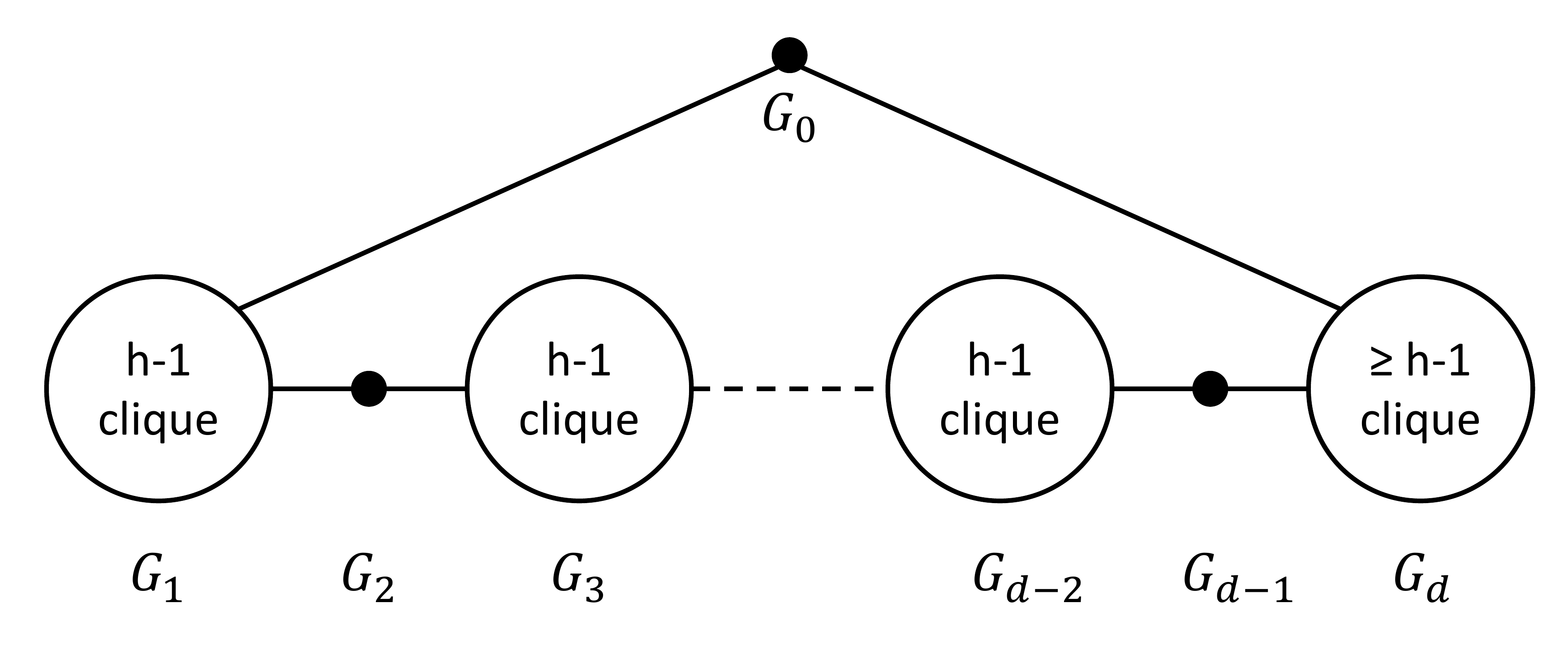}
    \caption{BRB Lower Bound: $(\lfloor \frac{n}{n-f} \rfloor -1)\Delta$ when $f\geq n/2$}
    \label{fig:BB-dishonest-majority-lb}
\end{figure}

\paragraph{BB upper bound $O(\frac{n}{n-f})\Delta$ under unsynchronized start and $f\geq n/2$.}

By the recent progress~\cite{constantroundBB} on the Byzantine broadcast protocol with expected constant round under $f\geq n/2$, we can directly obtain an upper bound on the \latency of $O(\frac{n}{n-f})\Delta$ for BB. For brevity we will omit the floor or ceiling on the accurate latency for the discussion below, as we don't have a tight bound for the $f\geq n/2$ case.
The BB protocol in~\cite{constantroundBB} is bootstrapped from a primitive called {\em TrustCast}, which takes about $\frac{2n}{n-f}$ rounds and can guarantee that each honest party either receives a message from the sender or knows the sender is Byzantine. 
The BB protocol is leader-based, and each epoch with the corresponding leader invokes $3$ instances of TrustCast, for the leader to send the proposal, the parties to vote for the proposal, and the parties to send commit certificate, respectively. 
When an honest leader is in charge, the protocol guarantees that all honest parties can commit after the voting (thus after the second TrustCast), and hence has \latency about $\frac{4n}{n-f}\Delta$.
More details of the expected constant round BB protocol can be found in~\cite{constantroundBB}.

Here we briefly describe how to further improve the upper bound of \latency to about $\frac{2n}{n-f}\Delta$ under $f\geq n/2$, based on the above BB protocol.
The idea is to add a fast path in the first round, where the broadcaster sends the proposal directly in $1$ round instead of invoking TrustCast for about $\frac{2n}{n-f}$ rounds, and then every party use TrustCast to send its vote, which takes about $\frac{2n}{n-f}$ rounds. Rest of the protocol such as the commit rule or the commit certificate remains the same.
If the broadcaster is honest, all honest parties can commit within about $\frac{2n}{n-f}\Delta$ time, and thus the \latency of the protocol is about $\frac{2n}{n-f}\Delta$.
Note that there is still a factor of $2$ gap between the lower bound and the upper bound for the $f\geq n/2$ case, which is an interesting open question for future work.







\end{document}